\documentclass[]{interact}

\usepackage{amsmath}
\usepackage{amsthm}
\usepackage{amssymb}
\usepackage{hyperref}
\usepackage{tikz}
\usepackage{subfigure}
\usepackage{stmaryrd}
\usepackage{booktabs}
\usepackage{mdframed}
\usepackage{upgreek}
\usepackage{underscore}
\usepackage{mathtools} \mathtoolsset{showonlyrefs=true} \MakeRobust{\eqref}
\usepackage{listings}
\usepackage{color}
\usepackage{fullpage}
\usepackage{stackrel}
\usepackage{textcomp}
\usepackage{enumerate}
\usepackage{accents}
\usepackage{graphicx}
\usepackage{epstopdf}
\usepackage{placeins}
\usepackage{bbm}
\usepackage{bm}
\usepackage{accents}
\usepackage[mathscr]{eucal}
\usepackage[linesnumbered,boxed,titlenumbered,ruled,nofillcomment]{algorithm2e}
\usepackage[T1]{fontenc}
\usepackage{lmodern}
\usepackage{xr}

\newtheorem{theorem}{Theorem}

\newtheorem{prop}[theorem]{Proposition}

\theoremstyle{definition}

\newtheorem{definition}[theorem]{Definition}

\newtheorem{example}[theorem]{Example}

\renewcommand{\(}{\left(}
\renewcommand{\)}{\right)}
\renewcommand{\[}{\left[}
\renewcommand{\]}{\right]}

\newcommand\Eb{\mathbb{E}}

\newcommand\Pb{\mathbb{P}}

\newcommand\Rb{\mathbb{R}}

\newcommand{\Abf}[0]{\mathbf{A}}
\newcommand{\Bbf}[0]{\mathbf{B}}
\newcommand{\Cbf}[0]{\mathbf{C}}
\newcommand{\Dbf}[0]{\mathbf{D}}
\newcommand{\Ebf}[0]{\mathbf{E}}

\newcommand{\Hbf}[0]{\mathbf{H}}
\newcommand{\Ibf}[0]{\mathbf{I}}
\newcommand{\Mbf}[0]{\mathbf{M}}
\newcommand{\Pbf}[0]{\mathbf{P}}

\newcommand{\Sbf}[0]{\mathbf{S}}
\newcommand{\Ubf}[0]{\mathbf{U}}
\newcommand{\Vbf}[0]{\mathbf{V}}
\newcommand{\Wbf}[0]{\mathbf{W}}

\newcommand{\abf}[0]{\mathbf{a}}
\newcommand{\bbf}[0]{\mathbf{b}}

\newcommand{\zbf}[0]{\mathbf{z}}

\newcommand{\Sc}[0]{\mathcal{S}}

\newcommand{\Ee}{\boldsymbol{\mathscr{E}}}

\newcommand{\Ue}{\boldsymbol{\mathscr{U}}}
\newcommand{\Ve}{\boldsymbol{\mathscr{V}}}
\newcommand{\We}{\boldsymbol{\mathscr{W}}}
\newcommand{\Xe}{\boldsymbol{\mathscr{X}}}
\newcommand{\Ye}{\boldsymbol{\mathscr{Y}}}

\newcommand{\diag}{\textnormal{diag}}

\newcommand{\fl}[0]{\operatorname{fl}}

\newcommand{\Uniform}[0]{\operatorname{Uniform}}
\newcommand{\op}[0]{\operatorname{op}}

\newcommand{\set}[1]{\left\{ #1 \right\}}

\newcommand{\rowvec}[1]{\begin{bmatrix} #1 \end{bmatrix}}

\newcommand{\defeq}{\stackrel{\text{\tiny def}}{=}}  

\usepackage{multicol}
\usepackage{blindtext}
\usepackage{caption}
\usepackage[utf8]{inputenc}
\usepackage[numbers,sort&compress]{natbib}

\allowdisplaybreaks

\hypersetup{pdfauthor={Osman Asif Malik, Stephen Becker},
	pdftitle={Randomization of Approximate Bilinear Computation for Matrix Multiplication},
	colorlinks=true,
	citecolor=blue,
	urlcolor=blue,
	linkcolor=blue,
}

\begin{document}

\title{Randomization of Approximate Bilinear Computation for Matrix Multiplication\thanks{The Version of Record of this manuscript has been published and is available in International Journal of Computer Mathematics: Computer Systems Theory 30 December 2020 \url{http://www.tandfonline.com/10.1080/23799927.2020.1861104}}}

\author{
	\name{Osman Asif Malik and Stephen Becker}
	\affil{Department of Applied Mathematics, University of Colorado Boulder, USA}
}

\maketitle

\begin{abstract}
	We present a method for randomizing formulas for bilinear computation of matrix products.
	We consider the implications of such randomization when there are two sources of error: One due to the formula itself only being approximately correct, and one due to using floating point arithmetic.  
	Our theoretical results and numerical experiments indicate that our method can improve performance when each of these error sources are present individually, as well as when they are present at the same time.
\end{abstract}


\section{Introduction} \label{sec:introduction}

Suppose $\Abf, \Bbf \in \Rb^{n \times n}$. In this paper, we are concerned with formulas for computing $\Cbf = \Abf \Bbf$ that take the form
\begin{equation} \label{eq:bilinear-computation}
	c_{ij} = \sum_{r=1}^R w_{ijr} \bigg(\sum_{k,l = 1}^n u_{k l r} a_{kl}\bigg) \bigg(\sum_{k',l' = 1}^n v_{k'l'r} b_{k'l'}\bigg),
\end{equation}
$i,j \in [n] \defeq \{1,2,\ldots,n\}$, where each $\Ue = (u_{k l r})$, $\Ve = (v_{k'l'r})$, and $\We = (w_{ijr})$ is a tensor containing real numbers, and $c_{ij}$ is the element at position $(i,j)$ in $\Cbf$, with similar notation for elements of $\Abf$ and $\Bbf$. Such a formula is called a \emph{bilinear computation} (BC) \citep{blaser2003}. We can rewrite the expression above as
\begin{equation} \label{eq:bilinear-computation-2}
c_{ij} = \sum_{k,l = 1}^n \sum_{k',l' = 1}^n a_{kl} b_{k'l'} \sum_{r=1}^R u_{klr} v_{k'l'r} w_{ijr}.
\end{equation}
Comparing this to the standard algorithm for matrix multiplication, we can see that for the computation \eqref{eq:bilinear-computation} to be exact, $\Ue$, $\Ve$ and $\We$ must satisfy
\begin{equation} \label{eq:requirement}
\sum_{r=1}^R u_{klr} v_{k'l'r} w_{ijr} = \delta_{ki} \delta_{l'j} \delta_{lk'} \;\;\;\; \text{for all } (k,l,k',l',i,j) \in [n]^6,
\end{equation}
where $\delta$ is the Kronecker delta (i.e., $\delta_{ki} = 1$ if $k=i$ and $0$ otherwise) \cite{brent1970}. If \eqref{eq:requirement} is satisfied, we say that \eqref{eq:bilinear-computation} is an \emph{exact bilinear computation} (EBC). The smallest positive integer $R$ for which there exist $\Ue$, $\Ve$ and $\We$ such that \eqref{eq:requirement} holds is referred to as the \emph{rank} of the computation. 
We refer to any algorithm of the form \eqref{eq:bilinear-computation} for which $R < n^3$ as \emph{fast}, since the asymptotic complexity is smaller than that of standard matrix multiplication, which has complexity $O(n^3)$.
Examples of fast computations of the form \eqref{eq:bilinear-computation} include Strassen's algorithm for $2 \times 2$ matrices \cite{strassen1969}, and Laderman's algorithm for $3 \times 3$ matrices \cite{laderman1976}. 
The formula \eqref{eq:bilinear-computation} is also valid if $\Abf$, $\Bbf$ and $\Cbf$ are of size $mn \times mn$ and $a_{kl}$, $b_{k'l'}$ and $c_{ij}$ are replaced by submatrices $\Abf_{kl}$, $\Bbf_{k'l'}$ and $\Cbf_{ij}$ of size $m \times m$, so that e.g.\ 
\begin{equation}
	\Abf = 
	\begin{bmatrix}
		\Abf_{1 1} & \Abf_{1 2} & \cdots & \Abf_{1 n} \\
		\Abf_{2 1} & \Abf_{2 2} & \cdots & \Abf_{2 n} \\
		\vdots     & \vdots     & 		 & \vdots 	  \\
		\Abf_{n 1} & \Abf_{n 2} & \cdots & \Abf_{n n} \\
	\end{bmatrix}.
\end{equation} 
This is why fast algorithms of the form \eqref{eq:bilinear-computation} can be used recursively to compute the product of larger matrices.

Define 6-way tensors 
\begin{align}
	\Xe &\defeq (\delta_{ki} \delta_{l'j} \delta_{lk'})_{(k,l,k',l',i,j) \in [n]^6}, \label{eq:def-x} \\
	\Ye &\defeq \big(\sum_{r=1}^R u_{klr} v_{k'l'r} w_{ijr}\big)_{(k,l,k',l',i,j) \in [n]^6}. \label{eq:def-y}
\end{align}
The condition in \eqref{eq:requirement} can be written succinctly as $\Ye = \Xe$. 
For both matrix and tensor inputs, let $\|\cdot\|$ denote the Frobenius norm, i.e., the square root of the sum of the square of all elements. 
We are interested in BCs that are only approximately correct, which motivates the following definition.
\begin{definition}[Approximate bilinear computation]
	Let $\Abf, \Bbf \in \Rb^{mn \times mn}$, where $m$ is a positive integer, and let $f$ be defined blockwise via
	\begin{equation} \label{eq:bilinear-computation-block}
		f(\Abf, \Bbf)_{ij} \defeq \sum_{r=1}^R w_{ijr} \bigg(\sum_{k,l = 1}^n u_{k l r} \Abf_{kl}\bigg) \bigg(\sum_{k',l' = 1}^n v_{k'l'r} \Bbf_{k'l'}\bigg) \in \Rb^{m \times m} \;\;\;\; \text{for } i,j \in [n],
	\end{equation}
	where $f(\Abf, \Bbf)_{ij}$, $\Abf_{kl}$ and $\Bbf_{k'l'}$ are submatrices of $f(\Abf,\Bbf)$, $\Abf$ and $\Bbf$, respectively, of size $m \times m$.
	We say that $f$ is an \emph{approximate bilinear computation} (ABC) with parameters $n$ and $\tau$, or $(n,\tau)$-ABC for short, if $\|\Ye - \Xe\| \leq \tau$. 
\end{definition}

We present a method for randomizing the computation \eqref{eq:bilinear-computation-block} and consider the implications of this approach when there are two kinds of error present.

\begin{itemize}
\item \textbf{Error due to an approximate algorithm.} Using an $(n, \tau)$-ABC with $\tau > 0$ will introduce error even if exact arithmetic is used. ABCs are interesting since it is sometimes hard or impossible to convert an approximate algorithm found numerically to an exact algorithm. Some examples of such approximate algorithms found numerically appear in \citep{smirnov2013}. Other papers that search for fast algorithms numerically include \citep{brent1970,johnson1986,benson2015,elser2016}.

\item \textbf{Numerical error due to using floating point arithmetic.} 
Although the standard algorithm for matrix multiplication also incurs numerical error, it is more severe in fast algorithms since the computation of an element $c_{ij}$ can involve elements from $\Abf$ and $\Bbf$ other than the vectors $\abf_{i:}$ and $\bbf_{:j}$, which are the $i$th row of $\Abf$ and $j$th column of $\Bbf$, respectively. In exact arithmetic, these additional terms cancel out for exact fast algorithms, but in finite precision those cancellations are typically not exact which can lead to substantial numerical error. These issues are exacerbated when the algorithm itself is only approximately correct. These considerations are especially important when using low precision environments, such as when computing on a GPU \citep{huang2018}. Low precision computation---using 32-bit, 16-bit, and even 8-bit precision numbers---is popular in, e.g., machine learning \citep{gupta2015, wang2018b, hopkins2019}.
\end{itemize}

We make the following contributions in this paper:
\begin{itemize}
	\item We propose a method for randomizing BCs for matrix multiplication which does not increase the leading order computational complexity of the algorithm.
	\item When exact arithmetic is used, we show that our randomized ABCs compute the correct matrix product in expectation. We also provide some performance guarantees.
	\item We show that these exact arithmetic results largely carry over to a setting when all computations are done in floating point arithmetic.
	\item When floating point arithmetic is used, we provide numerical evidence that randomizing EBCs using our scheme can reduce numerical error and improve robustness to adversarial examples.
\end{itemize}

\subsection{Related work} \label{sec:related-work}

Bini et al.~\citep{bini1979, bini1980} introduce a concept similar to ABC called Arbitrary Precision Approximating (APA) algorithms for matrix multiplication. 
For an APA algorithm, the tensor $\Xe$ and $\Ye$ defined in \eqref{eq:def-x} and \eqref{eq:def-y} satisfy the relationship
\begin{equation}
	\Ye(\varepsilon) + \Ee(\varepsilon) = \Xe,
\end{equation}
where 
\begin{equation} \label{eq:APA-y}
	\Ye(\varepsilon) \defeq \big(\sum_{r=1}^R u_{klr}(\varepsilon) v_{k'l'r}(\varepsilon) w_{ijr}(\varepsilon) \big)_{(k,l,k',l',i,j) \in [n]^6}
\end{equation}
is a function of $\varepsilon$, and $\Ee$ represents the error in the approximation $\Ye \approx \Xe$. 
Each entry in $\Ee$ is assumed to be a polynomial with zero constant coefficient, hence for every entry $\Ee_{klk'l'ij}(0) = 0$.
Consequently, an APA algorithm can be made arbitrarily accurate by making $\varepsilon$ small enough.
Moreover, as shown in \citep{bini1980b}, an EBC can be derived from an APA algorithm (see also the related discussion in Section~15.2 of \citep{burgisser2013}). 
This EBC takes the form
\begin{equation} \label{eq:APA-to-EBC}
	\sum_{i=1}^{d+1} \alpha_i \Ye(\varepsilon_i) = \Xe, 
\end{equation}
where all $\varepsilon_i$ are distinct, $d := \max_{klk'l'ij \in [n]^6} \deg(\Ee_{klk'l'ij}(\varepsilon))$ is the maximum degree of the polynomials describing the entries in $\Ee$, and where $[\alpha_1, \, \ldots, \,\alpha_{d+1}]$ is chosen as the solution to a certain linear system.
Bini~\citep{bini1980b} uses this approach to convert the APA scheme for $12 \times 12$ matrix multiplication in \citep{bini1979} to an exact one, which can be used recursively to get a matrix multiplication algorithm with complexity $O(n^{2.7799})$.
Fixing the error in an APA algorithm yields an ABC with some parameter $\tau$.
In some of our numerical experiments, we use an ABC which we get by fixing $\varepsilon$ in the $12 \times 12$ APA scheme of \citep{bini1979}.
We also consider an associated EBC derived via \eqref{eq:APA-to-EBC} in other experiments.
Our paper focuses on ABCs instead of APA algorithms since, in practice, it may be difficult or impossible to express a pair $(\Ye,\Ee)$ found numerically in terms of polynomials in $\varepsilon$.

To the best of our knowledge, our paper is the first to consider randomization as a tool for improving BCs for matrix multiplication which are only approximately correct. 
Various randomized algorithms for the standard matrix multiplication algorithm have been considered in other works; see e.g.\ \cite{drineas2006b, pagh2013}. 
For EBCs in floating point arithmetic, a patent by Castrapel and Gustafson \cite{castrapel2007} describes a randomized version of Strassen's algorithm which they claim reduces numerical error. 
They provide empirical support for this, but no mathematical proof. 
Our method generalizes their approach by randomly choosing from a wider range of equivalent algorithms. 
Additionally, our method can be applied to any formula of the form \eqref{eq:bilinear-computation}. 

Early works that analyze the stability of fast algorithms for matrix multiplication include \citep{brent1970a, bini1980, bini1980a}. 
Other works, such as \citep{dumitrescu1998,kaporin2004,demmel2007,ballard2016,desilva2018}, attempt to improve the stability of fast algorithms using other approaches that do not rely on randomization. 
In Section~\ref{sec:setting-2-experiments}, we compare our proposed method to the rescaling method in \citep{ballard2016}. 
In these experiments, we also consider a restricted version of our method which corresponds to the method in \citep{castrapel2007}.

\section{Randomization of bilinear computation for matrix multiplication} \label{sec:randomization-approx-matrix-mult}

In Section~\ref{sec:exact-arithmetic}, we first consider a setting in which exact arithmetic is used. In Section~\ref{sec:finite-arithmetic}, we then consider a setting in which floating point arithmetic is used.

\subsection{In exact arithmetic} \label{sec:exact-arithmetic}

We present our randomization scheme in the setting when the input matrices are $mn \times mn$ with $m \geq 1$\footnote{We present results for square matrices, but we believe the results can be extended to rectangular matrices.}.
Accordingly, suppose $f : \Rb^{mn \times mn} \times \Rb^{mn \times mn} \rightarrow \Rb^{mn \times mn}$ is an $(n,\tau)$-ABC. 
We will now define a randomized version of $f$, denoted by $\hat{f}$, which has the following property: 
For all $\Abf, \Bbf \in \Rb^{mn \times mn}$, $\Eb[\hat{f}(\Abf, \Bbf)] = \Abf \Bbf$.
To that end, let $\{s_i(j)\}_{(i,j) \in [3] \times [n]}$ be a collection of i.i.d.\ Rademacher random variables, i.e., each satisfying $\Pb [s_i(j) = +1] = \Pb [s_i(j) = -1] = 1/2$. Moreover, let $\pi_i : [n] \rightarrow [n]$, $i \in [3]$, be independent random permutation functions, each satisfying $(\forall (j,k) \in [n]^2)$ $\Pb[\pi_i(j) = k] = 1/n$. Let $\Sbf_i \in \Rb^{mn \times mn}$, $i \in [3]$, be block diagonal matrices with the $j$th nonzero block equal to $s_i(j) \Ibf_m$, where $\Ibf_m$ is the $m \times m$ identity matrix. Also, let $\Pbf_i \in \Rb^{mn \times mn}$, $i \in [3]$, be permutation matrices divided into $m \times m$ blocks, with blocks on position $(\pi_i(j), j)$, $j \in [n]$, equal to $\Ibf_m$ and all other blocks equal to zero. Define $\Mbf_i = \Pbf_i \Sbf_i$, $i \in [3]$. Note that each $\Mbf_i$ is orthogonal (i.e., $\Mbf_i^{-1} = \Mbf_i^\top$). We propose the following definition of $\hat{f}$.
\begin{definition}[Randomized approximate bilinear computation] \label{def:RandABC}
	Let $f$ be an $(n,\tau)$-ABC.
	We define a corresponding \emph{randomized approximate bilinear computation} with parameters $n$, $\tau$ and $\kappa$, or $(n, \tau, \kappa)$-RandABC for short, via
	\begin{equation} \label{eq:f-hat-def}
		\hat{f}(\Abf, \Bbf) \defeq 
		(1 - \kappa)^{-1} \Mbf_1^\top f(\Mbf_1 \Abf \Mbf_2^\top, \Mbf_2 \Bbf \Mbf_3^\top) \Mbf_3,
	\end{equation}
	where
	\begin{equation} \label{eq:epsilon-def}
		\kappa \defeq \frac{1}{n^3} \sum_{(i,j,l) \in [n]^3} \bigg(1 - \sum_{r=1}^R u_{ilr} v_{ljr} w_{ijr}\bigg)
	\end{equation}
	is assumed to satisfy $\kappa \neq 1$.
\end{definition}
Observe that if $f$ was an exact algorithm, then $\kappa = 0$ and $\hat{f}(\Abf, \Bbf) = \Abf \Bbf$ since the $\Mbf_i$'s would cancel out due to orthogonality. Since the cost of applying $\Mbf_i$ to an $mn \times mn$ matrix is $O(m^2 n^2)$, computing $\hat{f}(\Abf, \Bbf)$ has the same leading order complexity as computing $f(\Abf, \Bbf)$.
\begin{prop} \label{prop:expectation}
	Let $\hat{f}$ be an $(n, \tau, \kappa)$-RandABC with $\tau \geq 0$ and $\kappa \neq 1$ . For all $\Abf, \Bbf \in \Rb^{mn \times mn}$ we have $\Eb[\hat{f}(\Abf, \Bbf)] = \Abf \Bbf$.
\end{prop}
\begin{proof}
	Let $\hat{\Cbf} \defeq \hat{f}(\Abf, \Bbf)$. Considering the $(i,j)$th block of $\hat{\Cbf}$, and going through some tedious but straightforward algebra, we get
	\begin{equation} \label{eq:C-hat}
		\hat{\Cbf}_{ij} 
		= \frac{s_1(i) s_3(j)}{1-\kappa} 
		\sum_{k,l=1}^n \sum_{k', l' = 1}^n s_1(k) s_2(l) s_2(k') s_3(l') \Abf_{kl} \Bbf_{k'l'}  
		\sum_{r=1}^R u_{\pi_1(k) \pi_2(l) r} v_{\pi_2(k') \pi_3(l') r} w_{\pi_1(i) \pi_3(j) r}.
	\end{equation}
	If we take the expectation of this equation with respect to the random variables $\{s_i(j)\}$, most terms will vanish: If $i = k$, $l = k'$ and $j = l'$ for a given term, then the product of the $s_i$'s will be 1; otherwise, the expectation of that term will be zero due to independence and the fact that each $\Eb[s_i(j)] = 0$. Consequently, we have 
	\begin{equation}
		\Eb [ \hat{\Cbf}_{ij} ] 
		= (1-\kappa)^{-1} \sum_{l=1}^n \Abf_{il} \Bbf_{lj} 
		\Eb \Big[ \sum_{r=1}^R u_{\pi_1(i) \pi_2(l) r} v_{\pi_2(l) \pi_3(j) r} w_{\pi_1(i) \pi_3(j) r} \Big]
		= \sum_{l=1}^n \Abf_{il} \Bbf_{lj} = (\Abf \Bbf)_{ij},
	\end{equation}
	where the second equality is true since \eqref{eq:epsilon-def} implies that
	\begin{equation}
		\Eb \Big [ \sum_{r=1}^R u_{\pi_1(i) \pi_2(l) r} v_{\pi_2(l) \pi_3(j) r} w_{\pi_1(i) \pi_3(j) r} \Big ] = 1 - \kappa.
	\end{equation}
\end{proof}

It may seem surprising that Proposition~\ref{prop:expectation} holds for \emph{any} $\tau \geq 0$, since this means that even an ABC with an arbitrarily large error will be correct in expectation once randomized as in Definition~\ref{def:RandABC}.
However, as we will see in Proposition~\ref{prop:guarantee}, in order to be able to guarantee a small error for \emph{any realization} of $\hat{f}(\Abf,\Bbf)$, $\tau$ also needs to be small. 

Define $B_\mu^{(n)} \defeq \set{\Mbf \in \Rb^{n \times n} : \|\Mbf\| \leq \mu}$ and $\eta \defeq (1-\kappa)^{-1} - 1 = O(\kappa)$. The following proposition provides performance guarantees for $f$ and $\hat{f}$. 
\begin{prop} \label{prop:guarantee}
	Consider matrices $\Abf, \Bbf \in B_\mu^{(mn)}$.
	\begin{enumerate}[(i)]
		\item If $f$ is an $(n, \tau)$-ABC, then $\|f(\Abf,\Bbf) - \Abf\Bbf\| \leq \mu^2 \tau$.
		\item If $\hat{f}$ is an $(n, \tau, \kappa)$-RandABC, then $\|\hat{f}(\Abf,\Bbf) - \Abf\Bbf\| \leq \mu^2 |\eta| \|\Ye\| + \mu^2 \tau \lesssim \mu^2 |\kappa| \|\Ye\| + \mu^2 \tau $.  
		\item Moreover, $\sup_{\Abf, \Bbf \in B_\mu^{(mn)}} \| \hat{f}(\Abf, \Bbf) - \Abf \Bbf \| \leq \sup_{\Abf, \Bbf \in B_\mu^{(mn)}} \| f(\Abf,\Bbf) - \Abf \Bbf \| + \mu^2 |\eta| \|\Ye\|$.
	\end{enumerate}
\end{prop}
\begin{proof}
	We have
	\begin{align}
		&\|f(\Abf, \Bbf) - \Abf \Bbf\|^2 
		= \sum_{i,j=1}^n \bigg\| \sum_{k,l=1}^n \sum_{k',l'=1}^n \Abf_{kl} \Bbf_{k'l'} (y_{klk'l'ij} - x_{klk'l'ij}) \bigg\|^2 \\
		& \leq \sum_{i,j=1}^n \bigg( \sum_{k,l=1}^n \sum_{k',l'=1}^n \|\Abf_{kl}\| \|\Bbf_{k'l'}\| |y_{klk'l'ij} - x_{klk'l'ij}| \bigg)^2 \\
		& \leq \sum_{i,j=1}^n \bigg( \sum_{k,l=1}^n \sum_{k',l'=1}^n (\|\Abf_{kl}\| \|\Bbf_{k'l'}\|)^2 \bigg) \bigg( \sum_{k,l=1}^n \sum_{k',l'=1}^n |y_{klk'l'ij} - x_{klk'l'ij}|^2 \bigg) \\
		& = \bigg( \sum_{k,l=1}^n \|\Abf_{kl}\|^2 \sum_{k',l'=1}^n \|\Bbf_{k'l'}\|^2 \bigg) \bigg( \sum_{i,j=1}^n \sum_{k,l=1}^n \sum_{k',l'=1}^n |y_{klk'l'ij} - x_{klk'l'ij}|^2 \bigg)\\
		& = \|\Abf\|^2 \|\Bbf\|^2 \|\Ye - \Xe\|^2 \leq \mu^4 \tau^2,
	\end{align}
	where the first inequality follows from first applying the triangle inequality and then using the sub-multiplicativity of the Frobenius norm, and the second inequality follows from applying the Cauchy--Schwarz inequality. This proves (i). Since the Frobenius norm is invariant under unitary transformations,
	\begin{equation} \label{eq:guarantee-pf-2}
		\|\hat{f}(\Abf, \Bbf) - \Abf \Bbf\| = \| (1-\kappa)^{-1} f(\tilde{\Abf}, \tilde{\Bbf}) - \tilde{\Abf} \tilde{\Bbf} \|,
	\end{equation}
	where $\tilde{\Abf} = \Mbf_1 \Abf \Mbf_2^\top$, $\tilde{\Bbf} = \Mbf_2 \Bbf \Mbf_3^\top$. Going through the same computations as in the proof of (i), but with the extra $(1-\kappa)^{-1}$ term, we therefore get
	\begin{equation}
		\|\hat{f}(\Abf, \Bbf) - \Abf \Bbf\| \leq \|\tilde{\Abf}\| \|\tilde{\Bbf}\| \|(1-\kappa)^{-1} \Ye - \Xe\|
		= \|\Abf\| \|\Bbf\| \|(1-\kappa)^{-1} \Ye - \Xe\| 
		\leq \mu^2 ( |\eta|\| \Ye \| + \tau)
	\end{equation}
	where the equality once again uses the unitary invariance of the Frobenius norm, and the last inequality follows from the definitions of $\eta$ and the triangle inequality.	
	This proves (ii).
	Fix $\Abf, \Bbf \in B_\mu^{(mn)}$. Applying the triangle inequality to \eqref{eq:guarantee-pf-2} and using the definition of $\eta$, we get 
	\begin{equation} \label{eq:guarantee-pf-3}
		\|\hat{f}(\Abf, \Bbf) - \Abf \Bbf\|	\leq \| f(\tilde{\Abf}, \tilde{\Bbf}) - \tilde{\Abf} \tilde{\Bbf} \|
		+ |\eta| \, \|f(\tilde{\Abf}, \tilde{\Bbf})\|.
	\end{equation}
	By doing computations almost identical to those in the proof of (i), we get the bound
	\begin{equation} \label{eq:guarantee-pf-4}
		\|f(\tilde{\Abf}, \tilde{\Bbf})\| \leq \|\tilde{\Abf}\| \|\tilde{\Bbf}\| \|\Ye\| \leq \mu^2 \|\Ye\|,
	\end{equation}
	since $\tilde{\Abf}, \tilde{\Bbf} \in B_\mu^{(mn)}$. Combining \eqref{eq:guarantee-pf-3} and \eqref{eq:guarantee-pf-4} and taking supremums appropriately proves (iii).
\end{proof}

Points (i) and (ii) in Proposition~\ref{prop:guarantee} provide absolute performance guarantees for $f$ and $\hat{f}$, respectively. Note that a tighter version of (i) holds, with $\|\Abf\|$ replaced by $\sqrt{\sum_{kl} \|\Abf_{kl}\|_2^2}$, where $\|\cdot\|_2$ is the spectral norm. We keep (i) in its current looser form to make it easier to compare to (ii) and (iii). Point (iii) shows that the worst case performance of $\hat{f}$ is no worse than that of $f$ plus a constant.
In fact, that constant, which also appears in (ii), can be bounded as follows. 
\begin{prop} \label{prop:constant-bound}
	If $|\kappa| \leq 1/2$, then
	\begin{equation}
		|\eta| \mu^2 \|\Ye\| \leq 2 \mu^2 (n^{-5/2} \tau + n^{-1}) \tau. 
	\end{equation}
\end{prop}
\begin{proof}
	Let $\zbf \in \Rb^{n^3}$ be the vector with elements $(\sum_{r=1}^R u_{ilr} v_{ljr} w_{ijr} - 1)$ for $(i,j,l) \in [n]^3$, i.e., containing the elements of $(\Ye - \Xe)$ in positions for which $\Xe$ has an entry 1 (the element order in $\zbf$ is irrelevant). Note that
	\begin{equation} \label{eq:epsilon-bound}
		|\kappa| = \bigg| n^{-3} \sum_{i=1}^{n^3} z_i \bigg| \leq n^{-3} \|\zbf\|_1 \leq n^{-3} \sqrt{n} \|\zbf\|
		\leq n^{-5/2} \|\Ye - \Xe\|,
	\end{equation} 
	where the first equality follows from \eqref{eq:epsilon-def}, the first inequality follows from the triangle inequality and the definition of the 1-norm, and the second inequality is a well known relation (see e.g.\ Equation~(2.2.5) in \cite{golub2013}). 
	Now, note that 
	\begin{equation} \label{eq:Ye-bound}
		\|\Ye\| \leq \|\Ye - \Xe\| + \|\Xe\| = \|\Ye - \Xe\| + n^{3/2}.
	\end{equation}
	Combining \eqref{eq:epsilon-bound}, \eqref{eq:Ye-bound}, the fact that $|\eta| \leq 2 |\kappa|$ when $|\kappa| \leq 1/2$, and the definition of $\tau$ gives us the desired bound.
\end{proof}

The upper bound $\mu^2 \tau$ in Proposition~\ref{prop:guarantee} (i) is also an upper bound to the $\sup$ term on the right hand side of the inequality in (iii) of the same proposition. Proposition~\ref{prop:constant-bound} shows that the size of the additional constant $|\eta| \mu^2 \|\Ye\|$ is not much larger than the bound that we already have on this $\sup$ term, and that it will be smaller than that bound if e.g.\ $n \geq 3$ and $\tau < 3^{3/2}/2$.

As is clear from the proof of Proposition~\ref{prop:expectation}, the constant $\kappa$ in \eqref{eq:epsilon-def} is used to rescale the computation in \eqref{eq:f-hat-def} so that the output of $\hat{f}$ is correct in expectation. It is important to note that $\kappa$, which can be both positive and negative, is \emph{not} a measure of error in the algorithm. Indeed, setting $R=1$ and all elements of $\Ue$, $\Ve$ and $\We$ to 1 would result in $\kappa = 0$, but this would clearly be a very poor algorithm. Although the corresponding randomized algorithm $\hat{f}$ would be correct in expectation, the error guarantees in Proposition~\ref{prop:guarantee} (i) and (ii) would be very poor, since $\tau$ would be large.

\subsubsection{A recursive algorithm for approximate bilinear computation} \label{sec:recursive-algorithm}

We can extend the result in Proposition~\ref{prop:expectation} to a recursive version of $\hat{f}$. We denote the recursive algorithm with $Q$ recursions for multiplication of $m n^Q \times m n^Q$ matrices by $\hat{F}^{(Q)}$. Let $\{s_i^{(q)}(j)\}_{(i,j,q) \in [3] \times [n] \times [Q]}$ be a collection of i.i.d.\ Rademacher random variables, and let $\pi_i^{(q)} : [n] \rightarrow [n]$, $(i,q) \in [3] \times [Q]$, be independent random permutation functions, each satisfying $(\forall (j,k) \in [n]^2)$ $\Pb[\pi_i^{(q)}(j) = k] = 1/n$. Let $\hat{F}^{(1)}$ be defined exactly as $\hat{f}$ in \eqref{eq:f-hat-def} but based on the random variables $\{s_i^{(1)}(j)\}_{(i,j) \in [3] \times [n]}$ and $\{\pi_i^{(1)}\}_{i \in [3]}$ and define $\hat{F}^{(q)} : \Rb^{m n^q \times m n^q} \times \Rb^{m n^q \times m n^q} \rightarrow \Rb^{m n^q \times m n^q}$, $q \in \{2, 3, \ldots, Q\}$, recursively via
\begin{equation} \label{eq:recursive-random-formula}
	(\hat{F}^{(q)}(\Abf, \Bbf))_{ij} \defeq (1-\kappa)^{-1} s_1^{(q)}(i) s_3^{(q)}(j)
	\sum_{r=1}^R w_{\pi_1^{(q)}(i) \pi_3^{(q)}(j) r} \hat{F}^{(q-1)} \( \Abf^{(q)}, \Bbf^{(q)} \),
\end{equation}
where
\begin{equation}
\begin{aligned}
	&\Abf^{(q)} \defeq \sum_{k,l=1}^n u_{\pi_1^{(q)}(k) \pi_2^{(q)}(l) r} s_1^{(q)}(k) s_2^{(q)}(l) \Abf_{kl}, \\
	&\Bbf^{(q)} \defeq \sum_{k',l'=1}^n v_{\pi_2^{(q)}(k') \pi_3^{(q)}(l') r} s_2^{(q)}(k') s_3^{(q)}(l') \Bbf_{k'l'},
\end{aligned}
\end{equation}
and $(\hat{F}^{(q)}(\Abf, \Bbf))_{ij}$ is the subblock of size $m n^{q-1}\times m n^{q-1}$ on position $(i,j)$. If $\Abf$ and $\Bbf$ are of size $p \times p$ but there is no integer $m$ such that $p = m n^Q$, e.g.\ if $p$ is prime, then we can simply pad the matrices appropriately \cite{huss-lederman1996}.  If the recursive formula \eqref{eq:recursive-random-formula} is used $Q$ times to compute the product of two $N \times N$ matrices, where $N \defeq m n^Q$, and $\hat{F}^{(0)} (\Abf^{(0)}, \Bbf^{(0)}) = \Abf^{(0)} \Bbf^{(0)}$ is computed via the standard matrix multiplication formula, the asymptotic cost of the recursive algorithm is $O(m^{3-\log_n R} N^{\log_n R})$. This is the same leading order complexity as the corresponding computation without any randomization. For example, if we insert $m = 1$, $n = 2$ and $R = 7$, we recover the asymptotic cost of Strassen's algorithm: $O(N^{\log_2 7}) \approx O(N^{2.81})$.
\begin{prop} \label{prop:expectation-recursive}
	For any positive integer $Q$ and for all $\Abf, \Bbf \in \Rb^{m n^Q \times m n^Q}$ we have $\Eb[\hat{F}^{(Q)} (\Abf, \Bbf)] = \Abf \Bbf$.
\end{prop}
\begin{proof}
	Note that the claim is true for $Q = 1$ due to Proposition~\ref{prop:expectation}. Now assume it is true for some $Q \geq 1$. We will show that it is also true for $Q + 1$, i.e., that $\Eb[\hat{F}^{(Q+1)}(\Abf, \Bbf)] = \Abf \Bbf$ for $\Abf, \Bbf \in \Rb^{m n^{Q+1} \times m n^{Q+1}}$. Let $\Sc$ denote the $\sigma$-algebra generated by the random variables $\{s_i^{(Q+1)}(j)\}_{(i,j) \in [3] \times [n]}$ and $\{\pi_i^{(Q+1)}(j)\}_{(i,j) \in [3] \times [n]}$, and let $\Eb_{\Sc}[ \; \cdot \; ] \defeq \Eb[ \; \cdot \mid \Sc]$. Then
	\begin{equation}
		\Eb [(\hat{F}^{(Q+1)} (\Abf, \Bbf))_{ij} ] = \Eb \[ \Eb_{\Sc}[ (\hat{F}^{(Q+1)} (\Abf, \Bbf))_{ij} ] \]
	\end{equation}
	due to the smoothing property of expectation (property 10 in \cite[p.~348]{resnick2014}),
	\begin{equation}
	\begin{aligned}
		&= \Eb \bigg [ (1-\kappa)^{-1} s_1^{(Q+1)}(i) s_3^{(Q+1)} (j) \sum_{r=1}^R w_{\pi_1^{(Q+1)}(i) \pi_3^{(Q+1)}(j) r} \\
		&\;\;\;\; \times \Eb_{\Sc} \Big [ \hat{F}^{(Q)} \Big(\sum_{k,l=1}^n u_{\pi_1^{(Q+1)}(k) \pi_2^{(Q+1)}(l) r} s_1^{(Q+1)}(k) s_2^{(Q+1)}(l) \Abf_{kl}, \\
		&\;\;\;\; \sum_{k',l'=1}^n v_{\pi_2^{(Q+1)}(k') \pi_3^{(Q+1)}(l') r} s_2^{(Q+1)}(k') s_3^{(Q+1)}(l') \Bbf_{k'l'} \Big )  \Big ] \bigg] \\
	\end{aligned}
	\end{equation}
	since each $s_i^{(Q+1)}(j)$ and $\pi_i^{(Q+1)}(j)$ is $\Sc$-measurable,
	\begin{equation}
	\begin{aligned}
		&= \Eb \bigg [ (1-\kappa)^{-1} s_1^{(Q+1)}(i) s_3^{(Q+1)} (j) \sum_{r=1}^R w_{\pi_1^{(Q+1)}(i) \pi_3^{(Q+1)}(j) r} \\
		& \;\;\;\; \times  \Big (\sum_{k,l=1}^n u_{\pi_1^{(Q+1)}(k) \pi_2^{(Q+1)}(l) r} s_1^{(Q+1)}(k) s_2^{(Q+1)}(l) \Abf_{kl} \Big ) \\
		& \;\;\;\; \times \Big ( \sum_{k',l'=1}^n v_{\pi_2^{(Q+1)}(k') \pi_3^{(Q+1)}(l') r} s_2^{(Q+1)}(k') s_3^{(Q+1)}(l') \Bbf_{k'l'} \Big)  \bigg] \\
	\end{aligned}
	\end{equation}
	due to the induction hypothesis and since all random variables $\{s_i^{(q)}(j)\}_{(i,j,q) \in [3] \times [n] \times [Q]}$ and $\{\pi_i^{(q)}(j)\}_{(i,j,q) \in [3] \times [n] \times [Q]}$ are independent of $\Sc$ (and using property 12 in \cite[pp.~349--350]{resnick2014}),
	\begin{equation}
	\begin{aligned}
		& = \Eb \bigg[ (1-\kappa)^{-1} s_1^{(Q+1)}(i) s_3^{(Q+1)} (j) \sum_{k,l=1}^n \sum_{k',l'=1}^n s_1^{(Q+1)}(k) s_2^{(Q+1)}(l) s_2^{(Q+1)}(k') s_3^{(Q+1)}(l') \Abf_{kl} \Bbf_{k'l'} \\
		& \;\;\;\; \times \sum_{r=1}^R u_{\pi_1^{(Q+1)}(k) \pi_2^{(Q+1)}(l) r} v_{\pi_2^{(Q+1)}(k') \pi_3^{(Q+1)}(l') r} w_{\pi_1^{(Q+1)}(i) \pi_3^{(Q+1)}(j) r} \bigg] \\
	\end{aligned}
	\end{equation}
	by reordering the terms,
	\begin{equation}
		= \sum_{l=1}^n \Abf_{il} \Bbf_{lj} = (\Abf \Bbf)_{ij}
	\end{equation}
	which follows by doing the same analysis as in the proof of Proposition~\ref{prop:expectation}. So $\Eb [ \hat{F}^{(Q+1)}(\Abf, \Bbf) ] = \Abf \Bbf$. The claim in Proposition~\ref{prop:expectation-recursive} now follows by induction.
\end{proof}

\subsection{In floating point arithmetic} \label{sec:finite-arithmetic}

In floating point arithmetic, we cannot achieve guarantees like those in Propositions~\ref{prop:expectation} and \ref{prop:expectation-recursive}. This is illustrated in the following example.
\begin{example} \label{ex:EBC}
Let $g : \Rb^{mn \times mn} \times \Rb^{mn \times mn} \rightarrow \Rb^{mn \times mn}$ be a function that computes matrix multiplication according to some EBC in floating point arithmetic. Let $\hat{g}$ be defined analogously to $\hat{f}$ in \eqref{eq:f-hat-def}, but in terms of $g$ instead of $f$. Since $g$ is exact, we have $\kappa = 0$. We will use $\hat{G}^{(Q)}$ to denote the recursive version of $\hat{g}$, defined analogously to $\hat{F}^{(Q)}$.
We will use $G^{(Q)}$ to denote the recursive version of $g$, defined analogously to $\hat{G}^{(Q)}$ but with each $s_i^{(q)}(j) = 1$ and each $\pi_i^{(q)}(j) = j$, i.e., with no randomness involved so that $G^{(Q)}$ is deterministic. We consider the same setup as in Section 2.7.10 of \cite{golub2013}: Let 
\begin{equation}
	\Abf = \Bbf = 
	\begin{bmatrix}
	0.99 & 0.0010 \\
	0.0010 & 0.99 
	\end{bmatrix},
\end{equation}
and suppose we are computing on a machine using 2-digit floating point arithmetic. Taking $g$ to be Strassen's algorithm computed on this machine, we get $\|g(\Abf, \Bbf) - \Abf \Bbf\| \approx 0.0286$. In the definition of $\hat{g}$, there are a total of 64 possible sign functions and 8 possible permutation functions. Each combination of these has the same probability of occurring, so we can readily compute $\Eb[\hat{g}(\Abf,\Bbf)]$. Doing this, we find that $\|\Eb[\hat{g}(\Abf, \Bbf)] - \Abf \Bbf\| \approx 0.0024$. So we cannot guarantee $\Eb[\hat{g}(\Abf, \Bbf)] = \Abf \Bbf$ in this finite precision setting. We are not even guaranteed to have $\|\Eb[\hat{g}(\Abf, \Bbf)] - \Abf \Bbf\| \leq \|g(\Abf, \Bbf) - \Abf \Bbf\|$: Let $\Abf^{(r)}$ and $\Bbf^{(r)}$ be the same as $\Abf$ and $\Bbf$, respectively, but with the order of the columns reversed, i.e., $\Abf^{(r)} \defeq \rowvec{\abf_{:2} & \abf_{:1}}$ and $\Bbf^{(r)} \defeq \rowvec{\bbf_{:2} & \bbf_{:1}}$. We then get $\|g(\Abf^{(r)}, \Bbf^{(r)}) - \Abf^{(r)} \Bbf^{(r)}\| \approx 0.0001$, but $\|\Eb[\hat{g}(\Abf^{(r)}, \Bbf^{(r)})] - \Abf^{(r)} \Bbf^{(r)}\| \approx 0.0024$. Despite this, randomization seems to work remarkably well in practice when an exact algorithm is computed in finite precision arithmetic, as we will see in the numerical experiments.
\end{example}

We now consider ABCs (which include EBCs as a special case) in finite precision arithmetic. As in \citep{golub2013}, we use $\fl(x)$ to denote the representation of $x \in \Rb$ as a floating point number, and $\fl(f(x))$ to denote the result of computing $f(x)$ in floating point arithmetic. When computing the latter, the algorithm used to compute $f$ matters. We make the standard assumption that $\fl(x \op y) = (x \op y) (1+\Delta)$ where $x$ and $y$ are floating point numbers, $\op$ is scalar addition, subtraction or multiplication, and $|\Delta| \leq \varepsilon_\text{machine}$, where $\varepsilon_\text{machine}$ is the machine epsilon or unit roundoff. For numerical summations, e.g.\ $\fl(\sum_{n=1}^N x_n)$, we assume that the summation is simply done sequentially, e.g.\
\begin{equation}
	\fl\Big(\sum_{n=1}^3 x_n\Big) = \fl\Big(\fl\big(\fl(x_1) + \fl(x_2)\big) + \fl(x_3)\Big).
\end{equation}
We first present a series of results in Propositions~\ref{prop:numerical-error-scalar-case}--\ref{prop:total-error-matrix-case-randomized} with the goal of understanding how algorithmic \emph{and} numerical error combined impact the deterministic and  randomized ABCs described by $f$ and $\hat{f}$, respectively. Throughout these results, we make the reasonable assumption that the input matrices $\Abf$ and $\Bbf$, as well as the tensors $\Ue$, $\Ve$ and $\We$ defining the BC, are already stored in floating point format, so that e.g.\ $\fl(\Abf) = \Abf$. Proposition~\ref{prop:numerical-error-scalar-case} provides an upper bound on the numerical error for the (approximate or exact) BC $f$ defined as in \eqref{eq:bilinear-computation}.
\begin{prop} \label{prop:numerical-error-scalar-case}
	Suppose $(4n+R-1) \varepsilon_\textup{machine} \leq 0.01$. For $\Abf, \Bbf \in B_\mu^{(n)}$ and an $(n,\tau)$-ABC $f$ computed according to \eqref{eq:bilinear-computation-block}, we have
	\begin{equation}
	\|\fl(f(\Abf, \Bbf)) - f(\Abf, \Bbf) \| \leq  1.01 (4n + R - 1)\sqrt{R} \varepsilon_\textup{machine} \mu^2 \sqrt{\sum_{r=1}^R \|\Ubf_{::r}\|^2 \|\Vbf_{::r}\|^2 \|\Wbf_{::r}\|^2},
	\end{equation}
	where e.g.\ $\Ubf_{::r} \in \Rb^{n \times n}$ is the $r$th frontal slice of $\Ue$, so that $\|\Ubf_{::r}\|^2 = \sum_{k,l} u_{klr}^2$.
\end{prop}
\begin{proof}
	Throughout this proof, constants $\gamma$, $\theta$, $\lambda$, $\tilde{\gamma}$, $\tilde{\theta}$, $\tilde{\lambda}$, $\phi$ and $\psi$ with subscripts are real numbers of magnitude less than or equal to $\varepsilon_\text{machine}$. Define
	\begin{equation}
	s_{kpr} \defeq \fl \bigg( \sum_{l=1}^p u_{klr} a_{kl} \bigg).
	\end{equation}
	Then
	\begin{equation}
	\begin{aligned}
	&s_{k1r} = \fl(u_{k1r} a_{k1}) = u_{k1r} a_{k1} (1 + \gamma_{k1r}), \\
	&s_{k2r} = \fl(s_{k1r} + u_{k2r} a_{k2}) = (s_{k1r} + u_{k2r} a_{k2} (1 + \gamma_{k2r}))(1+\theta_{k2r}),
	\end{aligned}
	\end{equation}
	etc. More generally,
	\begin{equation}
	s_{kpr} = \sum_{l=1}^p u_{klr} a_{kl} (1 + \gamma_{klr}) \prod_{\alpha = l}^p (1+\theta_{k\alpha r}), \;\;\;\; \theta_{k 1 r} \defeq 0.
	\end{equation}
	Next, define
	\begin{equation}
	\hat{s}_{pr} \defeq \fl \bigg( \sum_{k=1}^p \sum_{l=1}^n u_{klr} a_{kl} \bigg).
	\end{equation}
	Then
	\begin{equation}
	\begin{aligned}
	&\hat{s}_{1r} = \fl \bigg( \sum_{l=1}^n u_{1lr} a_{1l} \bigg) = s_{1nr}, \\
	&\hat{s}_{2r} = \fl \bigg( \hat{s}_{1r} + \fl\Big( \sum_{l=1}^n u_{2lr} a_{2l} \Big)  \bigg) = (s_{1nr} + s_{2nr}) (1 + \lambda_{2r}),
	\end{aligned}
	\end{equation}
	etc. More generally,
	\begin{equation}
	\hat{s}_{pr} = \sum_{k=1}^p s_{knr} \prod_{\alpha=k}^p (1+\lambda_{\alpha r}), \;\;\;\; \lambda_{1 r} \defeq 0.
	\end{equation}
	Consequently,
	\begin{equation}
	P_r \defeq \fl \bigg( \sum_{k=1}^n \sum_{l=1}^n u_{klr} a_{kl} \bigg) = \hat{s}_{nr} = \sum_{k=1}^n \sum_{l=1}^n u_{klr} a_{kl} (1 + \gamma_{klr}) \bigg( \prod_{\alpha = l}^n (1+\theta_{k\alpha r}) \bigg) \bigg( \prod_{\alpha=k}^n (1+\lambda_{\alpha r}) \bigg).
	\end{equation}
	Similarly,
	\begin{equation}
	Q_r \defeq \fl \bigg( \sum_{k'=1}^n \sum_{l'=1}^n v_{k'l'r} b_{k'l'} \bigg) = \sum_{k'=1}^n \sum_{l'=1}^n v_{k'l'r} b_{k'l'} (1 + \tilde{\gamma}_{k'l'r}) \bigg( \prod_{\alpha = l'}^n (1+\tilde{\theta}_{k'\alpha r}) \bigg) \bigg( \prod_{\alpha=k'}^n  (1+\tilde{\lambda}_{\alpha r}) \bigg).
	\end{equation}
	Now, define
	\begin{equation}
	z_{p} \defeq \fl \bigg( \sum_{r=1}^p w_{ijr} \bigg(\sum_{k,l = 1}^n u_{k l r} a_{kl}\bigg) \bigg(\sum_{k',l' = 1}^n v_{k'l'r} b_{k'l'}\bigg) \bigg).
	\end{equation}
	We have
	\begin{equation}
	\begin{aligned}
	z_1 &= \fl \Big( w_{ij1} \Big(\sum_{k,l = 1}^n u_{k l 1} a_{kl}\Big) \Big(\sum_{k',l' = 1}^n v_{k'l'1} b_{k'l'}\Big) \Big) \\
		&= w_{ij1} P_1 Q_1 (1+\phi_{11}) (1+\phi_{12}), \\
	z_2 &= \fl \bigg(z_1 + \fl \Big( w_{ij2} \Big(\sum_{k,l = 1}^n u_{k l 2} a_{kl}\Big) \Big(\sum_{k',l' = 1}^n v_{k'l'2} b_{k'l'}\Big) \Big)\bigg) \\
		&= (z_1 + w_{ij2} P_2 Q_2 (1+\phi_{21}) (1+\phi_{22})) (1+\psi_{2}),
	\end{aligned}
	\end{equation}
	etc. From this, it follows that
	\begin{equation}
	\fl(f(\Abf, \Bbf)_{ij}) = z_R = \sum_{r=1}^R w_{ijr} P_r Q_r (1+\phi_{r1}) (1+\phi_{r2}) \prod_{\zeta = r}^R (1+\psi_{\zeta}), \;\;\;\; \psi_1 \defeq 0.
	\end{equation}
	Writing this out in full and rearranging, we have
	\begin{equation}
	\begin{aligned}
	&\fl(f(\Abf, \Bbf)_{ij}) = \sum_{k,l=1}^n \sum_{k',l'=1}^n a_{kl} b_{k'l'} \sum_{r=1}^R u_{klr} v_{k'l'r} w_{ijr} (1 + \gamma_{klr}) (1 + \tilde{\gamma}_{k'l'r})  \\
	&\times \bigg( \prod_{\alpha = l}^n (1+\theta_{k\alpha r}) \bigg) \bigg( \prod_{\alpha = l'}^n (1+\tilde{\theta}_{k'\alpha r}) \bigg) \bigg( \prod_{\alpha = k}^n (1+\lambda_{\alpha r}) \bigg) \\
	&\times \bigg( \prod_{\alpha = k'}^n(1+\tilde{\lambda}_{\alpha r}) \bigg) (1+\phi_{r1}) (1+\phi_{r2}) \prod_{\zeta = r}^R (1+\psi_{\zeta}).
	\end{aligned}
	\end{equation}
	Since we have assumed $(4n+R-1) \varepsilon_\text{machine} \leq 0.01$, it follows from Lemma~2.7.1 in \cite{golub2013} that there exist constants $\varepsilon_{klk'l'r}$, for $(k,l,k',l',r) \in [n]^4 \times [R]$, such that $|\varepsilon_{klk'l'r}| \leq 1.01 (4n + R - 1) \varepsilon_\text{machine}$ and
	\begin{equation}
	\begin{aligned}
	\fl(f(\Abf, \Bbf)_{ij}) 
	&= \sum_{k,l=1}^n \sum_{k',l'=1}^n a_{kl} b_{k'l'} \sum_{r=1}^R u_{klr} v_{k'l'r} w_{ijr} (1 + \varepsilon_{klk'l'r}) \\
	&= f(\Abf, \Bbf)_{ij} + \sum_{k,l=1}^n \sum_{k',l'=1}^n a_{kl} b_{k'l'} \sum_{r=1}^R u_{klr} v_{k'l'r} w_{ijr} \varepsilon_{klk'l'r}.
	\end{aligned}
	\end{equation}
	We have
	\begin{equation}
	\begin{aligned}
		&|\fl(f(\Abf, \Bbf)_{ij}) - f(\Abf, \Bbf)_{ij}|^2 = \Big| \sum_{k,l=1}^n \sum_{k',l'=1}^n a_{kl} b_{k'l'} \sum_{r=1}^R u_{klr} v_{k'l'r} w_{ijr} \varepsilon_{klk'l'r} \Big|^2 \\
		&\leq \Big(\sum_{k,l=1}^n \sum_{k',l'=1}^n (a_{kl} b_{k'l'})^2 \Big) \Big( \sum_{k,l=1}^n \sum_{k',l'=1}^n \Big( \sum_{r=1}^R u_{klr} v_{k'l'r} w_{ijr} \varepsilon_{klk'l'r} \Big)^2 \Big) \\
		&\leq \Big(\sum_{k,l=1}^n \sum_{k',l'=1}^n (a_{kl} b_{k'l'})^2 \Big) \Big( \sum_{k,l=1}^n \sum_{k',l'=1}^n \Big( \sum_{r=1}^R (u_{klr} v_{k'l'r} w_{ijr})^2 \Big) \Big( \sum_{r=1}^R \varepsilon_{klk'l'r}^2 \Big) \Big) \\
		&\leq \|\Abf\|^2 \|\Bbf\|^2 R (1.01 (4n + R - 1) \varepsilon_\text{machine})^2 \sum_{r=1}^R \sum_{k,l=1}^n \sum_{k',l'=1}^n (u_{klr} v_{k'l'r} w_{ijr})^2,
	\end{aligned}
	\end{equation}
	where the first and second inequality follows from the Cauchy--Schwarz inequality, and the third inequality follows from the bound on $|\varepsilon_{klk'l'r}|$.
	Summing over all $i,j$, taking a square root, and using that $\|\Abf\|, \|\Bbf\| \leq \mu$, we arrive at 
	\begin{equation}
	\|\fl(f(\Abf, \Bbf)) - f(\Abf, \Bbf) \| \leq 1.01 (4n + R - 1)\sqrt{R} \varepsilon_\text{machine} \mu^2 \sqrt{\sum_{r=1}^R \|\Ubf_{::r}\|^2 \|\Vbf_{::r}\|^2 \|\Wbf_{::r}\|^2}.
	\end{equation}
\end{proof}

Proposition~\ref{prop:numerical-error-matrix-case} generalizes Proposition~\ref{prop:numerical-error-scalar-case} to the case when the input matrices are of size $mn \times mn$.
\begin{prop} \label{prop:numerical-error-matrix-case}
	Suppose $(4n + m + R - 2) \varepsilon_\textup{machine} \leq 1.01$. For $\Abf, \Bbf \in B_\mu^{(mn)}$ and an $(n,\tau)$-ABC $f$ computed according to \eqref{eq:bilinear-computation-block}, we have 
	\begin{equation}
		\|\fl(f(\Abf, \Bbf)) - f(\Abf, \Bbf)\| \leq 1.01 (4n + m + R - 2) \sqrt{R} \varepsilon_\textup{machine} \mu^2 \sqrt{\sum_{r=1}^R \|\Ubf_{::r}\|^2 \|\Vbf_{::r}\|^2 \|\Wbf_{::r}\|^2}.
	\end{equation}
\end{prop}
\begin{proof}
	Note that all computations when calculating $\sum_{k,l} u_{klr} \Abf_{kl}$ and $\sum_{k',l'} v_{k'l'r} \Bbf_{k'l'}$ are done elementwise. The error accumulated for each entry will therefore be the same as in the case when $\Abf_{kl}$ and $\Bbf_{k'l'}$ are scalars. More precisely,
	\begin{equation} \label{eq:mat-error}
		\fl \bigg( \sum_{k,l=1}^n u_{klr} \Abf_{kl} \bigg) = \sum_{k,l=1}^n u_{klr} \Abf_{kl} \circledast \Ebf_{klr},
	\end{equation}
	where $\circledast$ is the Hadamard (elementwise) product, and $\Ebf_{klr} \in \Rb^{m \times m}$ is a matrix with entries which are the product of at most $2n-1$ terms of the form $(1+\varepsilon)$ with $|\varepsilon| \leq \varepsilon_\text{machine}$. A similar statement is true for $\fl(\sum_{k',l'=1}^n v_{k'l'r} \Bbf_{k'l'})$. When computing each
	\begin{equation} \label{eq:fl-matrix-computation}
	\fl \bigg( w_{ijr} \bigg(\sum_{k,l=1}^n u_{klr} \Abf_{kl} \bigg) \bigg( \sum_{k',l'=1}^n v_{k'l'r} \Bbf_{k'l'} \bigg) \bigg),
	\end{equation}
	we now have additional error compared to the scalar case due to the inner product computations. In order to avoid cumbersome notation, we do not write out the following computations in full. Using the result in Section~2.7.6 of \cite{golub2013} for the rounding error in inner products, Equation \eqref{eq:mat-error}, Lemma~2.7.1 in \citep{golub2013}, and letting $(f(\Abf, \Bbf)_{ij})_{\alpha \beta}$ denote the element on position $(\alpha,\beta)$ in matrix block $(i,j)$, we can compute
	\begin{equation}
		\fl((f(\Abf, \Bbf)_{ij})_{\alpha \beta}) = (f(\Abf, \Bbf)_{ij})_{\alpha \beta} + \sum_{r=1}^R w_{ijr} \sum_{k,l=1}^n \sum_{k',l'=1}^n \sum_{z=1}^m u_{klr} v_{k'l'r} (\Abf_{kl})_{\alpha z} (\Bbf_{k'l'})_{z \beta} \Theta_{klk'l'r \alpha \beta z},
	\end{equation}
	where each $|\Theta_{klk'l'r \alpha \beta z}| \leq \hat{\theta} \defeq 1.01 (4n + m + R - 2) \varepsilon_\text{machine}$. Rearranging this, we have 
	\begin{equation}
	\begin{aligned}
		&|\fl((f(\Abf, \Bbf)_{ij})_{\alpha \beta}) - (f(\Abf, \Bbf)_{ij})_{\alpha \beta}|^2  \\
		&= \Big| \sum_{r=1}^R \sum_{k,l=1}^n \sum_{k',l'=1}^n \sum_{z=1}^m u_{klr} v_{k'l'r} w_{ijr} (\Abf_{kl})_{\alpha z} (\Bbf_{k'l'})_{z \beta} \Theta_{klk'l'r \alpha \beta z} \Big|^2 \\
		&\leq \hat{\theta}^2 \Big| \sum_{r=1}^R \sum_{k,l=1}^n \sum_{k',l'=1}^n |u_{klr} v_{k'l'r} w_{ijr}| \sum_{z=1}^m |(\Abf_{kl})_{\alpha z} (\Bbf_{k'l'})_{z \beta}|  \Big|^2
	\end{aligned}
	\end{equation}
	due to the triangle inequality and definition of $\hat{\theta}$,
	\begin{align}
		& \leq R \hat{\theta}^2 \sum_{r=1}^R \Big| \sum_{k,l=1}^n \sum_{k',l'=1}^n |u_{klr} v_{k'l'r} w_{ijr}| \sum_{z=1}^m |(\Abf_{kl})_{\alpha z} (\Bbf_{k'l'})_{z \beta}|  \Big|^2 \\
		& \leq R \hat{\theta}^2 \sum_{r=1}^R \Big(\sum_{k,l=1}^n \sum_{k',l'=1}^n |u_{klr} v_{k'l'r} w_{ijr}|^2 \Big) \Big(\sum_{k,l=1}^n \sum_{k',l'=1}^n \Big| \sum_{z=1}^m | (\Abf_{kl})_{\alpha z} (\Bbf_{k'l'})_{z \beta}| \Big|^2 \Big) \\
		& \leq R \hat{\theta}^2 \sum_{r=1}^R \Big(\sum_{k,l=1}^n \sum_{k',l'=1}^n |u_{klr} v_{k'l'r} w_{ijr}|^2 \Big) \Big(\sum_{k,l=1}^n \sum_{k',l'=1}^n \Big( \sum_{z=1}^m | (\Abf_{kl})_{\alpha z}|^2 \Big) \Big(\sum_{z=1}^m |(\Bbf_{k'l'})_{z \beta}|^2 \Big) \Big) \\
		&= R (1.01 (4n + m + R - 2) \varepsilon_\text{machine})^2 \\
		&\;\;\;\;\times \Big(\sum_{k,l=1}^n \sum_{z=1}^m (\Abf_{kl})_{\alpha z}^2 \Big) \Big(\sum_{k',l'=1}^n \sum_{z=1}^m (\Bbf_{k'l'})_{z \beta}^2 \Big) \Big(\sum_{r = 1}^R \sum_{k,l=1}^n \sum_{k',l'=1}^n u_{klr}^2 v_{k'l'r}^2 w_{ijr}^2 \Big)
	\end{align}
	where each inequality follows from an application of the Cauchy--Schwarz inequality and the final equality follows from a rearrangement of terms.
	Finally, since
	\begin{equation}
		\|\fl(f(\Abf, \Bbf)) - f(\Abf, \Bbf)\|^2 = \sum_{\alpha, \beta = 1}^m \sum_{i,j=1}^n |\fl((f(\Abf, \Bbf)_{ij})_{\alpha \beta}) - (f(\Abf, \Bbf)_{ij})_{\alpha \beta}|^2 
	\end{equation}
	and $\|\Abf\|, \|\Bbf\| \leq \mu$, it follows that
	\begin{equation}
		\|\fl(f(\Abf, \Bbf)) - f(\Abf, \Bbf)\| \leq 1.01 (4n + m + R - 2) \sqrt{R} \varepsilon_\text{machine} \mu^2 \sqrt{\sum_{r=1}^R \|\Ubf_{::r}\|^2 \|\Vbf_{::r}\|^2 \|\Wbf_{::r}\|^2}.
	\end{equation}
\end{proof}

Proposition~\ref{prop:numerical-error-matrix-case-randomized} shows how the error in Proposition~\ref{prop:numerical-error-matrix-case} changes when a randomized algorithm $\hat{f}$ is used instead of its deterministic counterpart. In particular, the bound is worse by only a factor $(1-\kappa)^{-1}$.
\begin{prop} \label{prop:numerical-error-matrix-case-randomized}
	Suppose $(4n + m + R - 2) \varepsilon_\textup{machine} \leq 1.01$ and $1-\kappa > 0$. For $\Abf, \Bbf \in B_\mu^{(mn)}$ and an $(n,\tau,\kappa)$-RandABC $\hat{f}$ computed according to \eqref{eq:f-hat-def}, we have
	\begin{equation}
	\begin{aligned}
		&\|\fl(\hat{f}(\Abf, \Bbf)) - \hat{f}(\Abf, \Bbf)\| \\
		&\leq 1.01 (4n + m + R - 2) \sqrt{R} \varepsilon_\textup{machine} (1-\kappa)^{-1} \mu^2 \sqrt{\sum_{r=1}^R \|\Ubf_{::r}\|^2 \|\Vbf_{::r}\|^2 \|\Wbf_{::r}\|^2}.
	\end{aligned}
	\end{equation}
\end{prop}
\begin{proof}
	From \eqref{eq:C-hat}, it follows that
	\begin{equation} \label{eq:randomized-algorithm}
	\begin{aligned}
	&\hat{f}(\Abf, \Bbf)_{ij} = (1-\kappa)^{-1} s_1(i) s_3(j) \sum_{r=1}^R w_{\pi_1(i) \pi_3(j) r} \\
	&\times \bigg( \sum_{k,l=1}^n u_{\pi_1(k) \pi_2(l) r} s_1(k) s_2(l) \Abf_{kl} \bigg)\bigg( \sum_{k',l'=1}^n v_{\pi_2(k') \pi_3(l') r} s_2(k') s_3(l') \Bbf_{k'l'} \bigg) \\
	&= \sum_{r=1}^R \hat{w}_{ijr} \bigg( \sum_{k,l=1}^n \hat{u}_{klr} \Abf_{kl} \bigg)\bigg( \sum_{k',l'=1}^n \hat{v}_{k'l'r} \Bbf_{k'l'} \bigg),
	\end{aligned}
	\end{equation}
	where 
	\begin{equation}
	\begin{aligned}
	&\hat{u}_{klr} \defeq u_{\pi_1(k) \pi_2(l) r} s_1(k) s_2(l), \\
	&\hat{v}_{k'l'r} \defeq v_{\pi_2(k') \pi_3(l') r} s_2(k') s_3(l'), \\
	&\hat{w}_{ijr} \defeq (1-\kappa)^{-1} s_1(i) s_3(j) w_{\pi_1(i) \pi_3(j) r}.		
	\end{aligned}
	\end{equation}
	Note that $\|\hat{\Ubf}_{::r}\| = \|\Ubf_{::r}\|$, $\|\hat{\Vbf}_{::r}\| = \|\Vbf_{::r}\|$ and $\| \hat{\Wbf}_{::r} \| = (1-\kappa)^{-1} \| \Wbf_{::r} \|$. Using this fact, and applying Proposition~\ref{prop:numerical-error-matrix-case} to the computation in \eqref{eq:randomized-algorithm}, gives us the desired result.
\end{proof}

Propositions~\ref{prop:total-error-matrix-case} and \ref{prop:total-error-matrix-case-randomized} give upper bounds for the total error, both due to algorithmic error and numerical rounding, for the deterministic and randomized algorithms, respectively.
\begin{prop} \label{prop:total-error-matrix-case}
	Suppose $(4n + m + R - 2) \varepsilon_\textup{machine} \leq 1.01$. For $\Abf, \Bbf \in B_\mu^{(mn)}$ and an $(n,\tau)$-ABC $f$ computed according to \eqref{eq:bilinear-computation-block}, we have
	\begin{equation} \label{eq:error-num-algorithm}
	\begin{aligned}
		&\|\fl(f(\Abf, \Bbf)) - \Abf \Bbf \| \\
		&\leq 1.01 (4n + m + R - 2) \sqrt{R} \varepsilon_\textup{machine} \mu^2 \sqrt{\sum_{r=1}^R \|\Ubf_{::r}\|^2 \|\Vbf_{::r}\|^2 \|\Wbf_{::r}\|^2} + \mu^2 \tau.
	\end{aligned}
	\end{equation}
\end{prop}
\begin{proof}
	Follows by applying the triangle inequality and using Propositions~\ref{prop:numerical-error-matrix-case} and \ref{prop:guarantee} (i).
\end{proof}

\begin{prop} \label{prop:total-error-matrix-case-randomized}
	Suppose $(4n + m + R - 2) \varepsilon_\textup{machine} \leq 1.01$. For $\Abf, \Bbf \in B_\mu^{(mn)}$ and an $(n,\tau,\kappa)$-RandABC $\hat{f}$ computed according to \eqref{eq:f-hat-def}, we have
	\begin{equation} \label{eq:error-randomized-num-algorithm}
	\begin{aligned}
		&\|\Eb[\fl(\hat{f}(\Abf, \Bbf))] - \Abf \Bbf \| \\
		&\leq 1.01 (4n + m + R - 2) \sqrt{R} \varepsilon_\textup{machine} (1-\kappa)^{-1} \mu^2 \sqrt{\sum_{r=1}^R \|\Ubf_{::r}\|^2 \|\Vbf_{::r}\|^2 \|\Wbf_{::r}\|^2}.
	\end{aligned}
	\end{equation}
\end{prop}
\begin{proof}
	Follows by taking expectations of the inequality in Proposition~\ref{prop:numerical-error-matrix-case-randomized} and applying Jensen's inequality \citep{schaefer1976}. 
\end{proof}
Note that Proposition~\ref{prop:total-error-matrix-case-randomized} implies that
\begin{equation}
\Eb[\fl(\hat{f}(\Abf, \Bbf))] \rightarrow \Abf \Bbf \;\;\;\; \text{as} \;\;\;\; \varepsilon_\text{machine} \rightarrow 0, 
\end{equation}
which means that as the numerical precision increases, $\Eb[\fl(\hat{f}(\Abf, \Bbf))]$ approaches $\Abf \Bbf$. This is consistent with Proposition~\ref{prop:expectation}.

To make the bounds in \eqref{eq:error-num-algorithm} and \eqref{eq:error-randomized-num-algorithm} easier to compare, using the fact that $(1-\kappa)^{-1} = 1 + \kappa + O(\kappa^2)$, we can rewrite \eqref{eq:error-randomized-num-algorithm} as 
\begin{equation} \label{eq:error-randomized-num-algorithm-2}
\begin{aligned}
	&\|\Eb[\fl(\hat{f}(\Abf, \Bbf))] - \Abf \Bbf \| \\
	&\leq 1.01 (4n + m + R - 2) \sqrt{R} \varepsilon_\text{machine} \mu^2 \sqrt{\sum_{r=1}^R \|\Ubf_{::r}\|^2 \|\Vbf_{::r}\|^2 \|\Wbf_{::r}\|^2} \\
	&+ 1.01 (4n + m + R - 2) \sqrt{R} \varepsilon_\text{machine} \kappa \mu^2 \sqrt{\sum_{r=1}^R \|\Ubf_{::r}\|^2 \|\Vbf_{::r}\|^2 \|\Wbf_{::r}\|^2} + O(\varepsilon_\text{machine} \kappa^2).
\end{aligned}
\end{equation}
The first term on the right hand side of \eqref{eq:error-num-algorithm} and \eqref{eq:error-randomized-num-algorithm-2} are identical. Recall from \eqref{eq:epsilon-bound} that $|\kappa| \leq n^{-5/2} \|\Ye - \Xe\| \leq n^{-5/2} \tau$. Consequently, if the quantity 
\begin{equation} \label{eq:square-root-of-sum}
	 \sqrt{\sum_{r=1}^R \|\Ubf_{::r}\|^2 \|\Vbf_{::r}\|^2 \|\Wbf_{::r}\|^2}
\end{equation}
is not too large, and $m$, $n$ and $R$ are of moderate size, the second term in \eqref{eq:error-randomized-num-algorithm-2} will be smaller than the second term in \eqref{eq:error-num-algorithm}, showing that the result in Proposition~\ref{prop:expectation} largely carries over to a setting with floating point arithmetic. The following example illustrates this.
\begin{example}
	For the sake of this example, suppose we are using an approximate variant of Strassen's algorithm. Then $R = 7$ and $n = 2$. For Strassen's algorithm, the quantity in \eqref{eq:square-root-of-sum} is approximately 6, so we will assume that it is less than 10 for our approximate variant of the algorithm. Suppose we are multiplying two $\text{100,000} \times \text{100,000}$ matrices in single precision, so that $m = \text{50,000}$ and $\varepsilon_\text{machine} \sim 10^{-8}$. Then, the second term in \eqref{eq:error-randomized-num-algorithm-2} is upper bounded by a quantity which is on the order of
	\begin{equation}
		10 \cdot 1.01 \cdot (4 \cdot 2+50000+7-2) \cdot \frac{\sqrt{7}}{2^{5/2}} \cdot 10^{-8} \mu^2 \tau \approx 0.0024 \mu^2 \tau,
	\end{equation}
	which is much smaller than the second term in \eqref{eq:error-num-algorithm}.
	The implementation of Strassen's algorithm in \citep{huang2018} achieves a speed-up over an efficient implementation of the standard matrix multiplication algorithm for square matrices with as few as 1,536 rows/columns. So multiplication of matrices with 100,000 rows/columns is well beyond the problem size for which fast algorithms can outperform the standard algorithm.
\end{example}

\section{Experiments} \label{sec:experiments}

In this section we present some results from experiments, with additional results provided in Appendix~\ref{app:additional-experiments}. We implement all experiments in Matlab with certain parts implemented in C. All our code is available online at \url{https://github.com/OsmanMalik/random-approximate-matrix-multiplication}.

In our experiments, we draw the matrices $\Abf, \Bbf$ from different random distributions. By \emph{Gaussian matrix}, we mean a matrix whose elements are realizations of i.i.d.\ standard normal random variables. Similarly, a \emph{uniform matrix} is one whose elements are realizations of i.i.d.\ $\Uniform(0,1)$ random variables. We also consider three types of random \emph{adversarial matrices}, which were proposed by \citep{ballard2016} and are designed to be challenging for Strassen's algorithm.
\begin{definition}[Adversarial matrices]
	Consider a matrix pair $\Abf, \Bbf \in \Rb^{n \times n}$. We say that it is \emph{type 1 adversarial} if
	\begin{equation}
	\begin{aligned}
	&a_{ij} \sim
	\begin{cases}
	\Uniform(0, 1/n^2) 	& \text{if } j > n/2, \\
	\Uniform(0,1) 		& \text{otherwise},
	\end{cases}
	\;\;\;\; 
	&b_{ij} \sim
	\begin{cases}
	\Uniform(0, 1/n^2) 	& \text{if } i < n/2, \\
	\Uniform(0,1) 		& \text{otherwise}.
	\end{cases}
	\end{aligned}
	\end{equation}
	We say that the matrix pair is \emph{type 2 adversarial} if 
	\begin{equation}
	\begin{aligned}
	&a_{ij} \sim
	\begin{cases}
	\Uniform(0, n^2) 	& \text{if } i < n/2 \text{ and } j > n/2, \\
	\Uniform(0,1) 		& \text{otherwise},
	\end{cases}
	\;\;\;\;
	&b_{ij} \sim
	\begin{cases}
	\Uniform(0, 1/n^2) 	& \text{if } j < n/2, \\
	\Uniform(0,1) 		& \text{otherwise}.
	\end{cases}
	\end{aligned}
	\end{equation}
	We say that the matrix pair is \emph{type 3 adversarial} if
	\begin{equation}
	a_{ij}, b_{ij} \sim
	\begin{cases}
	\Uniform(0, 1/n^2) 	& \text{if } i < n/2 \text{ and } j > n/2, \text{ or if } i \geq n/2 \text{ and } j \leq n/2 \\
	\Uniform(0,1) 		& \text{otherwise},
	\end{cases}
	\end{equation}
	Here, all the entries are assumed to be independent.
\end{definition}
We will also consider the Hilbert matrix $\Hbf \in \Rb^{n \times n}$, which has entries $h_{ij} = 1/(i+j-1)$, since it appears to be a particularly challenging matrix to the rescaling method which we consider in Section~\ref{sec:setting-2-experiments}.

For an ABC, the corresponding randomized computation with $Q$ recursions, $\hat{F}^{(Q)}$, was defined in \eqref{eq:recursive-random-formula}. In this section, we will use $F^{(Q)}$ to denote the deterministic counterpart. It is defined in the same way, but with each $s_i^{(q)}(j) = 1$ and each $\pi_i^{(q)}(j) = j$, i.e., with no randomness involved. As in Example~\ref{ex:EBC}, for EBCs we will use $\hat{G}^{(Q)}$ and $G^{(Q)}$ to denote the corresponding randomized and deterministic computations with $Q$ recursions.

\subsection{Approximate algorithm} \label{sec:setting-1-experiments}

We create an ABC of the form \eqref{eq:bilinear-computation-block} by taking the tensors $\Ue$, $\Ve$ and $\We$ corresponding to Strassen's algorithm and perturbing them: For each of the three tensors, we add i.i.d.\ mean zero Gaussian noise with standard deviation $10^{-3}$ to each element in the tensor equal to 1 as well as to five randomly selected elements that are equal to 0. Since we do the computations in double precision, and since we add a considerable amount of noise, these experiments are designed to test Propositions~\ref{prop:expectation}--\ref{prop:expectation-recursive}, i.e., we can ignore any floating point error.

In the first experiment, we draw two Gaussian matrices $\Abf, \Bbf \in \Rb^{80 \times 80}$ and compute 
\begin{equation} \label{eq:error-of-average}
	\Big\| \frac{1}{n} \sum_{i=1}^n \hat{F}^{(Q)}_i(\Abf, \Bbf) - \Abf \Bbf \Big\| / \|\Abf \Bbf\|
\end{equation}
for $n \in [10^4]$ and $Q \in [3]$. Here, $\hat{F}^{(Q)}_i$ is the $i$th realization of $\hat{F}^{(Q)}$. Figure~\ref{fig:Experiment2} shows the results. As expected from Propositions~\ref{prop:expectation} and \ref{prop:expectation-recursive}, the quantity in \eqref{eq:error-of-average} becomes smaller as $n$ increases.  

In the second experiment, we draw two Gaussian matrices $\Abf, \Bbf \in \Rb^{320 \times 320}$ and compute
\begin{equation} \label{eq:average-of-error}
	\| \hat{F}^{(Q)}_i(\Abf, \Bbf) - \Abf \Bbf \| / \|\Abf \Bbf\|
\end{equation}
for $i \in [100]$ and $Q \in [5]$, i.e., the relative error for each of 100 trials. The box plots in Figure~\ref{fig:Experiment3} compares the empirical distribution of \eqref{eq:average-of-error} to the relative error for the deterministic approximate algorithm, i.e., $\|F^{(Q)}(\Abf, \Bbf) - \Abf \Bbf\| / \|\Abf \Bbf\|$. In this particular case, randomization does not impact the median error and there is very little variation between trials. Figure~\ref{fig:Experiment3-hilbert} repeats this experiment, but with $\Abf = \Bbf = \Hbf$, where $\Hbf$ is the $320 \times 320$ Hilbert matrix. In this case, the randomized scheme frequently results in a slightly larger error. 
In Figure~\ref{fig:S-experiment3-mat-type-normal}--\ref{fig:S-experiment3-mat-type-adversarial-3} in Appendix~\ref{app:additional-experiments} we provide additional results for when $\Abf, \Bbf$ are Gaussian, uniform and type 1--3 adversarial. Those results demonstrate that randomization can both increase and decrease the error in this setting. However, as expected from Propositions~\ref{prop:guarantee} and \ref{prop:constant-bound}, the difference in error between the randomized and deterministic variants is typically not substantial.

\begin{figure}[ht!]
	\begin{center}
		\includegraphics[width=.6\columnwidth]{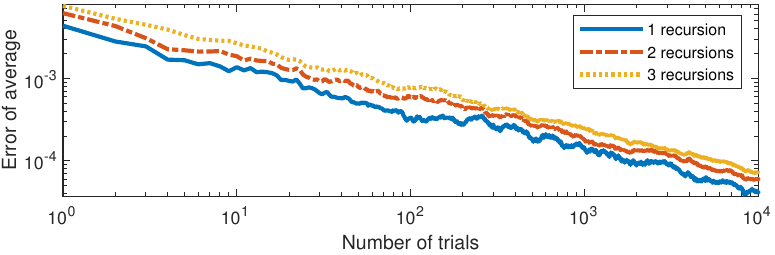}
	\end{center}
	\caption{Error of average for randomized ABC. $\Abf, \Bbf \in \Rb^{80 \times 80}$ are Gaussian and remain fixed throughout the experiment. The ABC is a perturbed variant of Strassen's algorithm.}
	\label{fig:Experiment2}
\end{figure}

\begin{figure}[ht!]
	\begin{center}
		\includegraphics[width=.6\columnwidth]{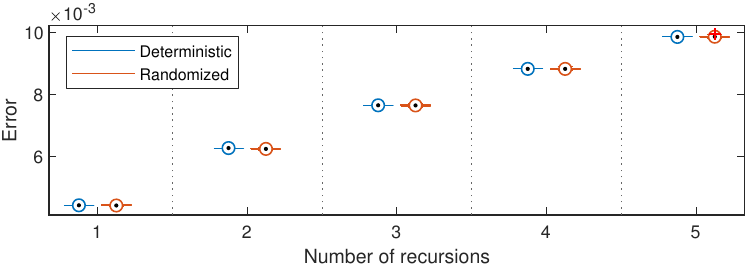}
	\end{center}
	\caption{Error for deterministic ABC compared to the error of the randomized counterpart, over 100 realizations of the randomized algorithm. $\Abf, \Bbf \in \Rb^{320 \times 320}$ are Gaussian and remain fixed over all realizations. Note that, unlike Figure~\ref{fig:Experiment2}, this figure shows the distribution of errors, rather than the errors of averages. The ABC is a perturbed variant of Strassen's algorithm.}
	\label{fig:Experiment3}
\end{figure}

\begin{figure}[ht!]
	\begin{center}
		\includegraphics[width=.6\columnwidth]{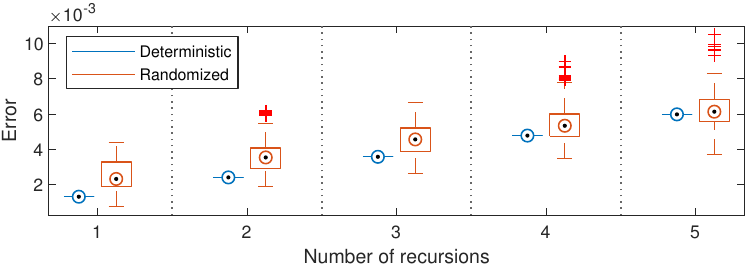}
	\end{center}
	\caption{Same as Figure~\ref{fig:Experiment3}, but with $\Abf$ and $\Bbf$ equal to the $320 \times 320$ Hilbert matrix.}
	\label{fig:Experiment3-hilbert}
\end{figure}

To see how these results carry over to other ABCs, we also present results in Figure~\ref{fig:Experiment3-bini} from experiments which use an instance of the $12 \times 12$ APA algorithm in \citep{bini1979} with $\varepsilon = \text{1e\textminus4}$ in the representation \eqref{eq:APA-y}. Results are shown for single recursion experiments on $12 \times 12$ matrices drawn from different distributions. As for the ABC based on the perturbed variant of Strassen's algorithm, the randomized variant of this ABC sometimes results in a smaller and sometimes in a larger error than the deterministic counterpart. 

\begin{figure}[ht!]
	\centering  
	\includegraphics[width=.8\textwidth]{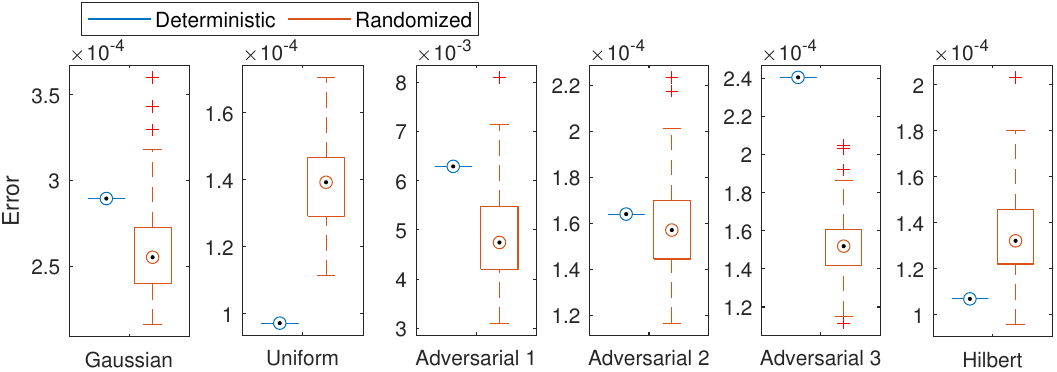}
	\caption{Error for deterministic ABC compared to the error of the randomized counterpart, over 100 realizations of the randomized algorithm. $\Abf, \Bbf \in \Rb^{12 \times 12}$ are drawn from six different random distributions, one in each subplot, and remain fixed over all realizations. The ABC is an instance of the $12 \times 12$ APA algorithm of \citep{bini1979}.}
	\label{fig:Experiment3-bini}
\end{figure}

\subsection{Exact algorithm in single precision floating point arithmetic} \label{sec:setting-2-experiments}

We first consider Strassen's algorithm, without any perturbations so that it is exact, when the computations are done in single precision floating point arithmetic. In error computations, we use the double precision product for $\Abf \Bbf$ computed using the standard algorithm as the true value of the product. Recall that by ``standard algorithm'' we mean the $O(n^3)$ algorithm.

In the first experiment, we draw two Gaussian matrices $\Abf, \Bbf \in \Rb^{80 \times 80}$ and compute the quantity in \eqref{eq:error-of-average}, but with each $\hat{F}_i^{(Q)}$ replaced by $\hat{G}_i^{(Q)}$, where $\hat{G}_i^{(Q)}$ is the $i$th realization of $\hat{G}^{(Q)}$, for $n \in [10^4]$ and $Q \in [3]$. Figure~\ref{fig:Experiment4} shows the results, where we also have included the error for the standard algorithm computed in single precision as a reference. 
Although it is clear that the randomized algorithms do not converge to the exact correct answer, it seems like their expectations perform better than the standard algorithm. 
Although the figure only shows the result for a specific random pair $(\Abf, \Bbf)$, we get qualitatively similar results every time we draw a new Gaussian matrix pair. The fact that the average of a single, or a few, outcomes of the randomized algorithms performs worse than the standard algorithm is to be expected, since Strassen's algorithm is more susceptible to numerical error that the standard algorithm. 
Figure~\ref{fig:Experiment4-hilbert} repeats this experiment, but with $\Abf = \Bbf = \Hbf$, where $\Hbf$ is the $80 \times 80$ Hilbert matrix. The results are similar to those in Figure~\ref{fig:Experiment4}.
Figures~\ref{fig:S-experiment4-mat-type-normal}--\ref{fig:S-experiment4-mat-type-adversarial-3} in the appendix provide additional results for when $\Abf, \Bbf$ are Gaussian, uniform and type 1--3 adversarial.

\begin{figure}[ht!]
	\begin{center}
		\includegraphics[width=.6\columnwidth]{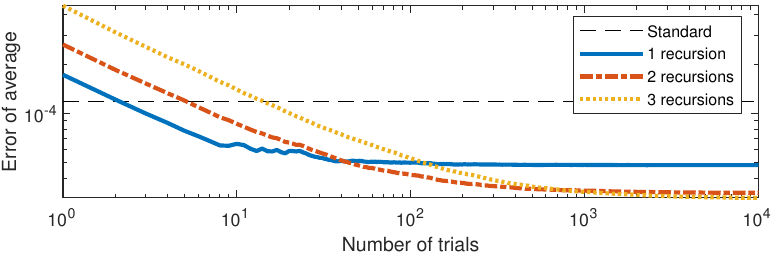}
	\end{center}
	\caption{Error of average for randomized EBC compared to the standard $O(n^3)$ algorithm in single precision floating point arithmetic. $\Abf, \Bbf \in \Rb^{80 \times 80}$ are Gaussian and remain fixed throughout the experiment.}
	\label{fig:Experiment4}
\end{figure}

\begin{figure}[ht!]
	\begin{center}
		\includegraphics[width=.6\columnwidth]{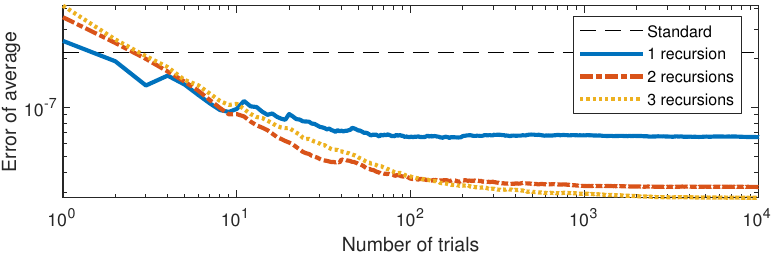}
	\end{center}
	\caption{Same as Figure~\ref{fig:Experiment4}, but with $\Abf$ and $\Bbf$ equal to the $80 \times 80$ Hilbert matrix.}
	\label{fig:Experiment4-hilbert}
\end{figure}

In the second experiment, we first compare the deterministic and three randomized versions of Strassen's algorithm, as well as the rescaling scheme proposed in \citep{ballard2016}. The first randomized version uses both random permutations and random signs. The two other randomized versions use only random signs and only random permutations, respectively. We include these two variants to better understand how random signs and random permutations each impact the performance. These three randomized methods will be referred to as ``fully randomized,'' ``random sign'' and ``random permutation,'' respectively. The algorithm that only uses random permutations corresponds to the method suggested in \cite{castrapel2007}.

The purpose of the rescaling scheme in \citep{ballard2016} is to improve numerical stability of fast EBCs, and has two steps: Outside scaling and inside scaling. With outside scaling, $\Cbf = \Abf \Bbf$ is computed via 
\begin{equation} \label{eq:outside-scaling}
	\Cbf_\text{outside} \defeq \Dbf_\Abf G^{(Q)}(\Dbf_\Abf^{-1} \Abf, \Bbf \Dbf_\Bbf^{-1}) \Dbf_\Bbf, 
\end{equation}
where $\Dbf_\Abf \defeq \diag (\max_{j} |a_{ij}|)$ and $\Dbf_\Bbf \defeq \diag (\max_{i} |b_{ij}|)$.
With inside scaling, $\Cbf$ is instead computed via
\begin{equation} \label{eq:inside-scaling}
	\Cbf_{\text{inside}} \defeq G^{(Q)}(\Abf \Dbf, \Dbf^{-1} \Bbf),
\end{equation}
where $\Dbf \defeq \diag(\sqrt{\max_j |b_{kj}|/ \max_i |a_{ik}|})$. 
In exact arithmetic, the scaling matrices in \eqref{eq:outside-scaling} and \eqref{eq:inside-scaling} will cancel out, in which case $\Cbf_\text{outside} = \Cbf_\text{inside} = \Cbf$. Both outside and inside scaling can be applied at the same time, as well as multiple times in an alternating fashion. The rescaling scheme that works best in numerical experiments in \citep{ballard2016} does outside-inside rescaling twice. We use the same rescaling scheme in our experiments, which we refer to as ``Rescaled 2x O-I.'' See Section~6 in \citep{ballard2016}, in particular Algorithm~3, for further details on the rescaling method.

For the experiment, we draw two random matrices $\Abf, \Bbf \in \Rb^{320 \times 320}$ and compute the quantity in \eqref{eq:average-of-error}, but with each $\hat{F}_i^{(Q)}$ replaced by $\hat{G}_i^{(Q)}$, for $i \in [100]$ and $Q \in [5]$. Figures~\ref{fig:Experiment5-normal}--\ref{fig:Experiment5-adversarial-3} show the results for Gaussian, uniform and type 1--3 adversarial matrices. Figure~\ref{fig:Experiment5-hilbert} shows the result when $\Abf$ and $\Bbf$ are both Hilbert matrices. These figures include the error of the standard algorithm computed in single precision as a reference. 

When the matrices are Gaussian (Figure~\ref{fig:Experiment5-normal}), all algorithms perform roughly the same with little variation between trials. 
For uniform matrices (Figure~\ref{fig:Experiment5-uniform}), the rescaling method performs about the same as the deterministic method. The fully randomized method has a lower error than the deterministic method, and it seems like this improvement comes from the random signs.
For type 1 adversarial matrices (Figure~\ref{fig:Experiment5-adversarial-1}), the rescaling method does remarkably well, achieving the same accuracy as the standard algorithm. The fully randomized algorithm outperforms the deterministic algorithm, and it seems like this improvement is coming from the random signs. 
For type 2 adversarial matrices (Figure~\ref{fig:Experiment5-adversarial-2}), the fully randomized method will sometimes perform worse than the deterministic algorithm, but has a lower median error for 2 or more recursions. The rescaling method also improves on the deterministic method, although the median error for the fully randomized method is lower for 4 recursions or more. 
Type 3 adversarial matrices were specifically proposed in \citep{ballard2016} to show a situation when rescaling does not work. This is clear in Figure~\ref{fig:Experiment5-adversarial-3}, where the rescaling method has the same error as the deterministic method. Our fully randomized method, however, has a lower error, and it seems like both the random signs and random permutations contribute to this performance improvement. When the matrices are Hilbert matrices (Figure~\ref{fig:Experiment5-hilbert}), the rescaling method does very poorly, with a much larger error than the deterministic method. Once again, our randomized method reduces the error compared to the deterministic method, with both the random signs and random permutations contributing to the improved performance.
Figures~\ref{fig:S-experiment5-mat-type-normal}--\ref{fig:S-experiment5-mat-type-adversarial-3} in the appendix provide additional results for when $\Abf, \Bbf$ are Gaussian, uniform, and type 1--3 adversarial. 

\begin{figure}[ht!]
	\begin{center}
		\includegraphics[width=.6\columnwidth]{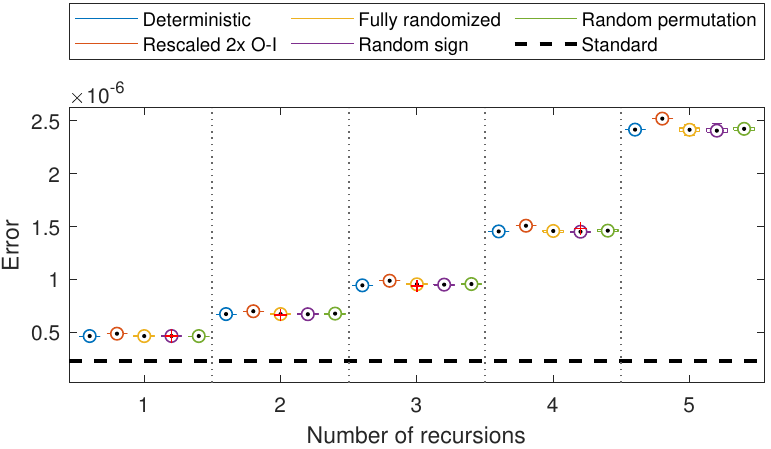}
	\end{center}
	\caption{Error for different variants of the Strassen EBC in single precision floating point arithmetic, over 100 realizations of the randomized algorithms. $\Abf, \Bbf \in \Rb^{320 \times 320}$ are Gaussian and remain fixed over all realizations.}
	\label{fig:Experiment5-normal}
\end{figure}

\begin{figure}[ht!]
	\begin{center}
		\includegraphics[width=.6\columnwidth]{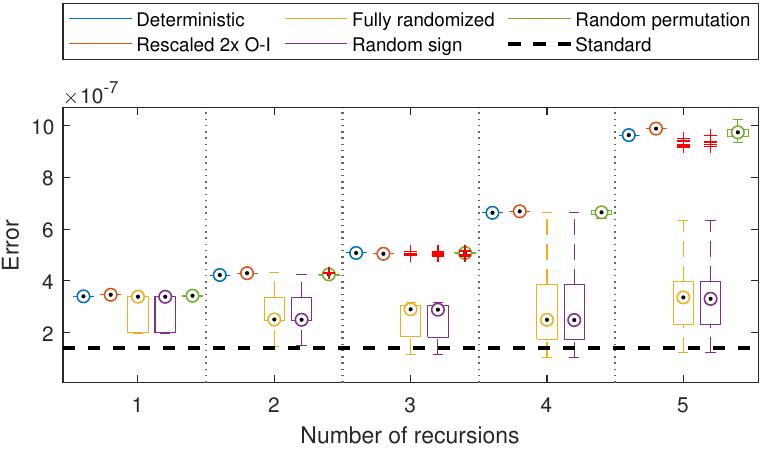}
	\end{center}
	\caption{Same as Figure~\ref{fig:Experiment5-normal}, but with uniform matrices.}
	\label{fig:Experiment5-uniform}
\end{figure}

\begin{figure}[ht!]
	\begin{center}
		\includegraphics[width=.6\columnwidth]{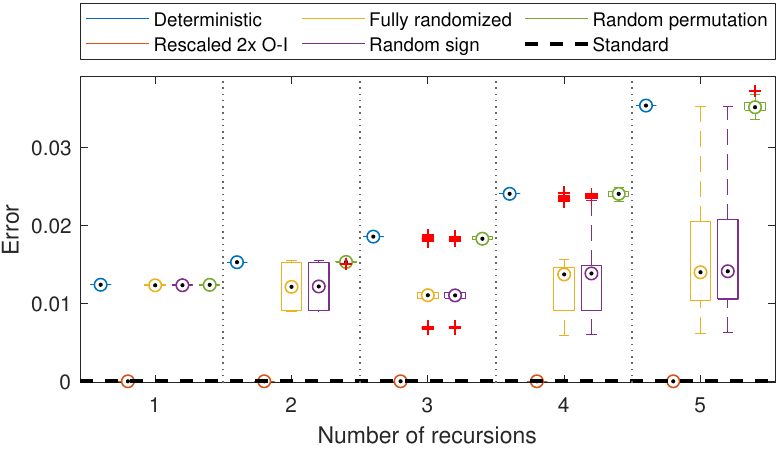}
	\end{center}
	\caption{Same as Figure~\ref{fig:Experiment5-normal}, but with type 1 adversarial matrices.}
	\label{fig:Experiment5-adversarial-1}
\end{figure}

\begin{figure}[ht!]
	\begin{center}
		\includegraphics[width=.6\columnwidth]{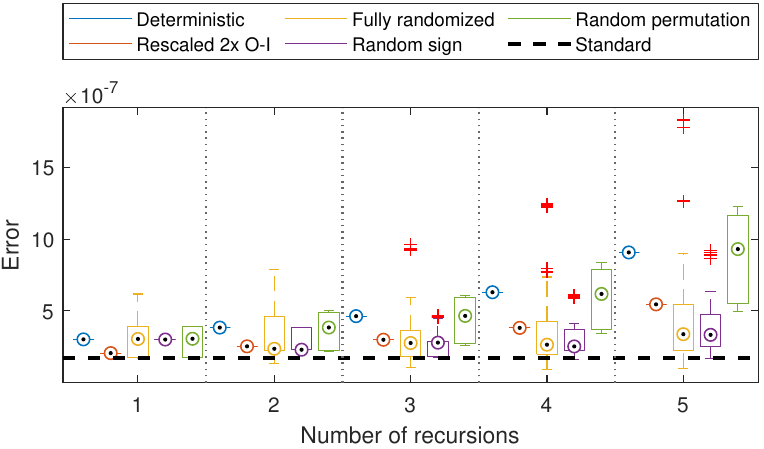}
	\end{center}
	\caption{Same as Figure~\ref{fig:Experiment5-normal}, but with type 2 adversarial matrices.}
	\label{fig:Experiment5-adversarial-2}
\end{figure}

\begin{figure}[ht!]
	\begin{center}
		\includegraphics[width=.6\columnwidth]{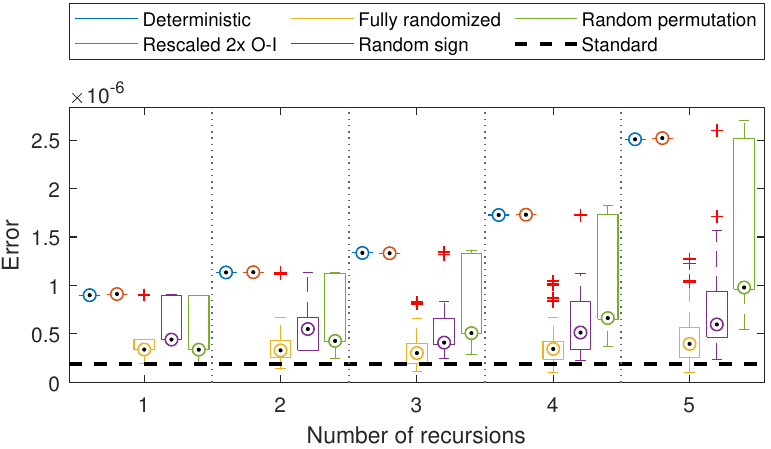}
	\end{center}
	\caption{Same as Figure~\ref{fig:Experiment5-normal}, but with type 3 adversarial matrices.}
	\label{fig:Experiment5-adversarial-3}
\end{figure}

\begin{figure}[ht!]
	\begin{center}
		\includegraphics[width=.49\textwidth]{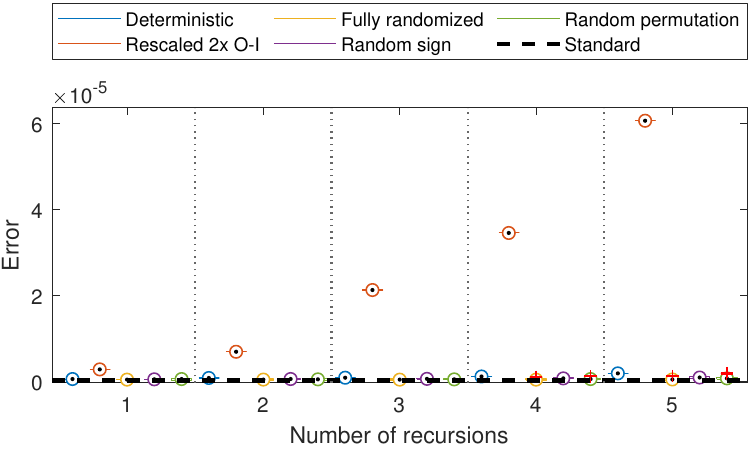}
		\includegraphics[width=.49\textwidth]{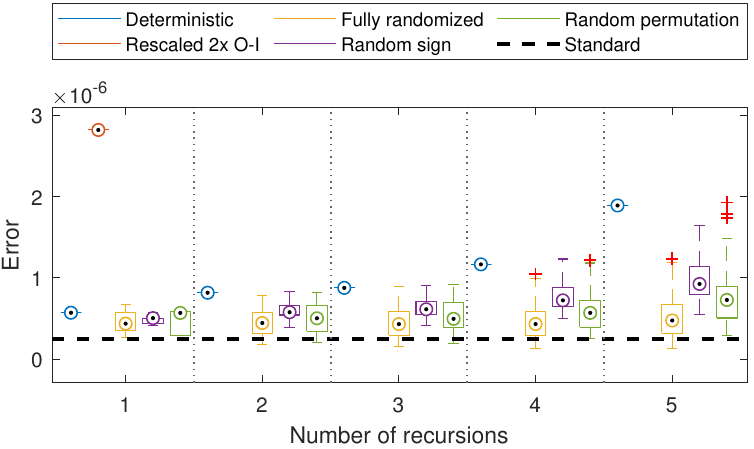}
	\end{center}
	\caption{Same as Figure~\ref{fig:Experiment5-normal}, but with Hilbert matrices. The left and right plots show the same results, with the right plot zoomed in closer to the interesting portion. The left plot is included to give a sense of the size of the errors for the rescaling method compared to the other methods.}
	\label{fig:Experiment5-hilbert}
\end{figure}

In Figure~\ref{fig:Experiment5-bini}, we give the results for some experiments which use an EBC for $12 \times 12$ matrix multiplication derived as in \citep{bini1980b} from the $12 \times 12$ APA algorithm in \citep{bini1979} via the approach discussed in Section~\ref{sec:related-work}.\footnote{There appears to be a few typos in the definition of $w_r^{(s)}$ in Equation~(5.2) in \citep{bini1980b}, which defines the APA scheme. We encourage the reader to consult our code for a corrected definition.} 
The APA algorithm in question has an error tensor whose entries are polynomials with maximum degree $d=6$.
Consequently, 7 distinct values of $\varepsilon$ are required in \eqref{eq:APA-to-EBC} to derive an exact scheme.
For this purpose, we choose $\varepsilon_1, \varepsilon_2, \ldots, \varepsilon_7$ to be $0.1, 0.2, \ldots, 0.7$. 
Our results are for single recursion experiments on $12 \times 12$ matrices drawn from different distributions. 
The results are similar to those for the Strassen EBC, except for the Gaussian case where the variability of the randomized algorithms is much higher. The variability of the randomized methods also appears to be somewhat higher on the other matrix types as well. Overall, our randomized methods perform favorably compared to the deterministic method, especially for the adversarial matrices and the Hilbert matrix. The rescaling method once again performs very well on type 1 adversarial matrices, but poorly on the Hilbert matrix.

\begin{figure}[ht!]
	\centering  
	\includegraphics[width=.8\textwidth]{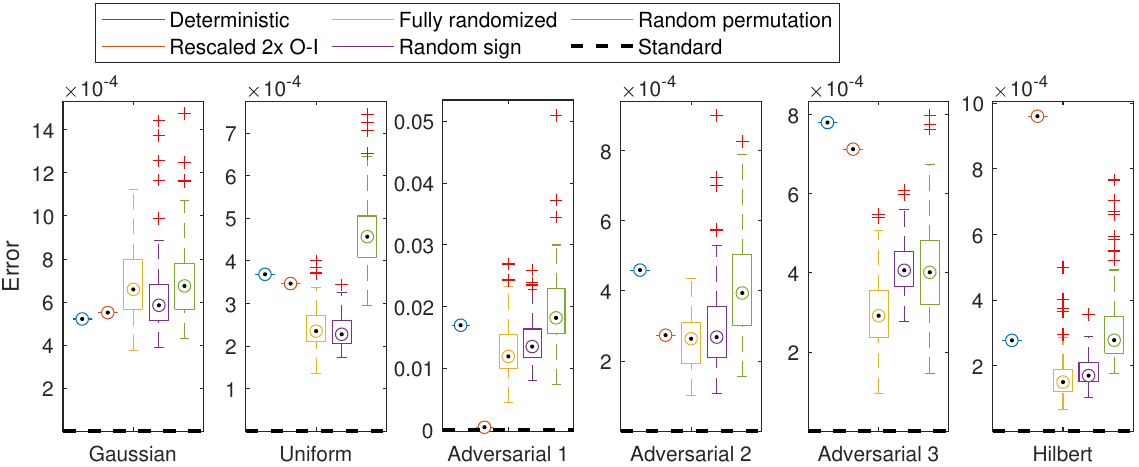}
	\caption{Error for different variants of the Bini~\citep{bini1980b} EBC in single precision floating point arithmetic, over 100 realizations of the randomized algorithms. $\Abf, \Bbf \in \Rb^{12 \times 12}$ are drawn from six different random distributions, one in each subplot, and remain fixed over all realizations.}
	\label{fig:Experiment5-bini}
\end{figure}

Although our method has some variability in performance due to being randomized, it seems to reduce the error compared to the deterministic method more reliably on a wider range of matrices than the rescaling method does, especially when more recursions are used. Our method also works well with the Hilbert matrix, which leads to large errors for the rescaling method. One benefit of our randomized method is that it can be done exactly even in low precision arithmetic, since it only involves permutation of rows/columns and the flipping of signs. The rescaling method, on the other hand, involves diagonal matrices with floating point numbers along the diagonals, which adds another potential source for numerical error in the algorithm. These experiments also indicate that both random signs and random permutations may separately help to reduce the error, and that combining the two seems to lower the error further.

\subsection{Randomization of ABC derived from APA algorithm}

In Section~\ref{sec:related-work} we discussed APA algorithms and how to derive EBCs from them by taking a linear combination of a few instances of the APA algorithm following an idea of Bini~\citep{bini1980b}.
Although the EBCs derived in this fashion are exact mathematically, they may suffer from greater numerical error than the standard $O(n^3)$ algorithm due to cancellation.
In this subsection, we compare an instance of the EBC in \citep{bini1980b}, which is derived from the APA algorithm in \citep{bini1979}, to a randomized version of the ABC we get by fixing the error parameter in the same APA algorithm.
The goal with these experiments is to see if the expectation of our randomized ABC can perform better than the EBC.
For the EBC, which is given by \eqref{eq:APA-to-EBC}, we choose $\varepsilon_1, \varepsilon_2, \ldots, \varepsilon_7$ to be $0.1, 0.2, \ldots, 0.7$. 
To get an ABC, we choose $\varepsilon = \text{1.5e\textminus2}$ in the APA algorithm. 
This ABC is then randomized as in Definition~\ref{def:RandABC}. 
The experiments are similar to that in Figure~\ref{fig:Experiment2}, but with an added baseline which shows the performance of the EBC.
The experiments are done for a single recursion on $12 \times 12$ matrices.
All computations are done in single precision arithmetic. 
Figure~\ref{fig:Experiment-6} shows the results, which each subplot showing the outcome for a different kind of matrix. 
Based on these result, the expectation of the randomized ABC seems to perform better than the EBC.
\begin{figure}[ht!]
	\centering  
	\includegraphics[width=.9\textwidth]{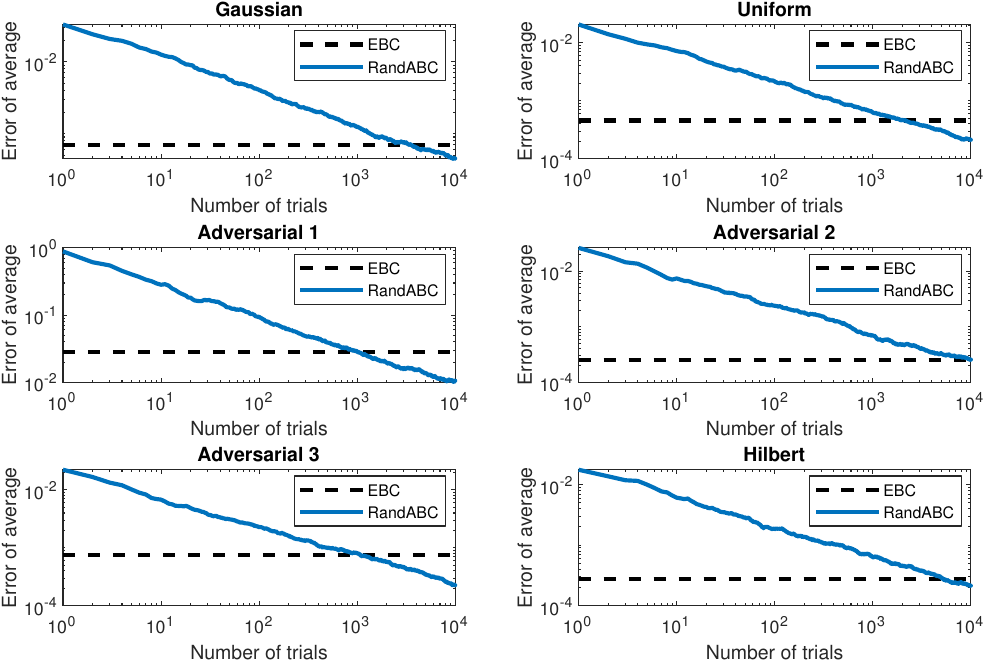}
	\caption{Comparison of EBC in \citep{bini1980b} to randomized ABC derived from APA algorithm in \citep{bini1979}.}
	\label{fig:Experiment-6}
\end{figure}

\section{Conclusion} \label{sec:conclusion}

In this paper we have suggested an approach for randomizing formulas for bilinear computation of matrix products which does not increase the asymptotic computational complexity. 
We have considered the implications of this approach when there are two sources of error: The first due to the algorithm itself being only approximately correct, and the second due to numerical error from using floating point arithmetic.
We believe that our results are encouraging, and provide ideas for improving the properties of matrix multiplication when these error sources are present separately or together. 

An interesting area for future research is to investigate other methods of randomization (e.g.\ combining the random sign changes in this paper with the fast Hadamard transform) and see if such a method can further improve the results. It would also be interesting to investigate how fast randomized approximate methods for matrix multiplication can be used in applications. One potential application areas is for computations in neural networks, where some amount of error in the computation is acceptable, and where randomization may help improve robustness of the trained model. It would also be interesting to investigate how our randomization scheme can be used as a component in fast linear algebra algorithms, such as recursive dense matrix inversion; see \citep{demmel2007a} for details.

\section*{Acknowledgments}

This material is based upon work supported by the National Science Foundation under Grant No.\ ECCS-1810314.

\bibliographystyle{tfs}
\bibliography{new-zotero-library}

\appendix

\section{Additional experiments} \label{app:additional-experiments}

First, we repeat the second experiment we did in Section~\ref{sec:setting-1-experiments} for the Strassen ABC with multiple Gaussian, uniform, and type 1--3 adversarial matrices. We use the same setup as in the main manuscript: For each experiment, we create an approximate algorithm by perturbing Strassen's algorithm, we draw random matrices $\Abf, \Bbf \in \Rb^{320 \times 320}$, and then we compute the quantity in (\ref{eq:average-of-error}) for $i \in [100]$ and $Q \in [5]$ and compare it to the relative error for the deterministic approximate algorithm. Figures~\ref{fig:S-experiment3-mat-type-normal}--\ref{fig:S-experiment3-mat-type-adversarial-3} show the results. In the case of Gaussian matrices (Figure~\ref{fig:S-experiment3-mat-type-normal}), the results look very similar to those in the main manuscript, with almost no difference in error between the deterministic and the randomized approximate algorithms. Figures~\ref{fig:S-experiment3-mat-type-uniform}--\ref{fig:S-experiment3-mat-type-adversarial-3} (uniform and type 1--3 adversarial) show that randomization can both increase and decrease the error. Note that both the deterministic and randomized approximate algorithms do particularly poorly on type 1 adversarial matrices.

Next, we repeat the first experiment in Section~\ref{sec:setting-2-experiments} with multiple Gaussian, uniform and type 1--3 adversarial matrices. We use the same setup as in the main manuscript: For each experiment, we draw random matrices $\Abf, \Bbf \in \Rb^{80 \times 80}$ and compute the quantity in \eqref{eq:error-of-average}, but with each $\hat{F}_i^{(Q)}$ replaced by $\hat{G}_i^{(Q)}$, for $n \in [10^4]$ and $Q \in [3]$. Figures~\ref{fig:S-experiment4-mat-type-normal}--\ref{fig:S-experiment4-mat-type-adversarial-3} show the results. Although using more recursions seems to increase the error for a single realization of the algorithm, it also seems to reduce the error of the expectation of the computed matrix product. 

Finally, we repeat the second experiment in Section~\ref{sec:setting-2-experiments} for the Strassen EBC with multiple Gaussian, uniform and type 1--3 adversarial matrices. We use the same setup as in the main manuscript: For each experiment, we draw random matrices $\Abf, \Bbf \in \Rb^{320 \times 320}$ and compute the quantity in \eqref{eq:average-of-error}, but with each $\hat{F}_i^{(Q)}$ replaced by $\hat{G}_i^{(Q)}$, for $i \in [100]$ and $Q \in [5]$. We do the same for variants of the algorithm that only use random permutations or random sign functions. We compare these to the deterministic version $G^{(Q)}$, and also include the error of the standard algorithm computed in single precision as a reference. Figures~\ref{fig:S-experiment5-mat-type-normal}--\ref{fig:S-experiment5-mat-type-adversarial-3} show the results, which look very similar to those presented in the main manuscript.

\begin{minipage}[ht!]{.47\textwidth}
	\includegraphics[width=1\textwidth]{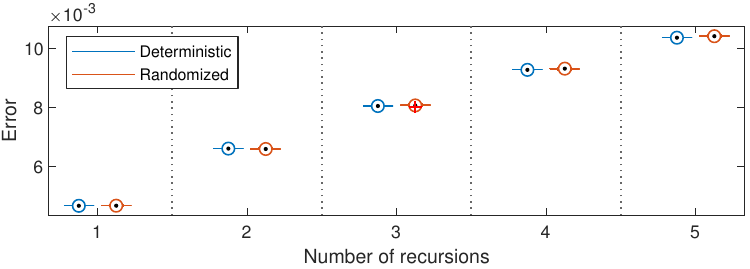}
	\includegraphics[width=1\textwidth]{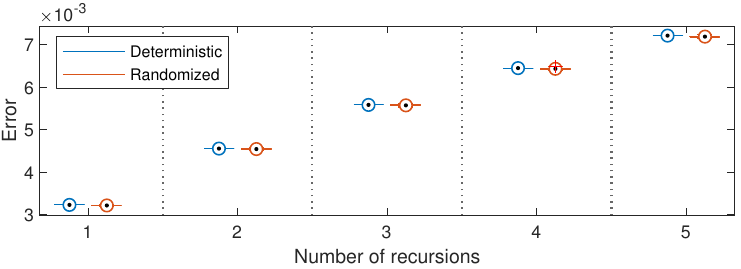}
	\includegraphics[width=1\textwidth]{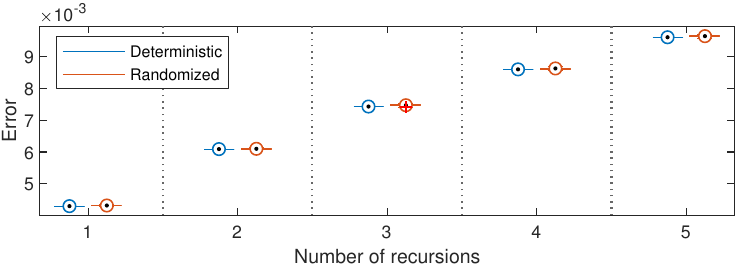}
	\includegraphics[width=1\textwidth]{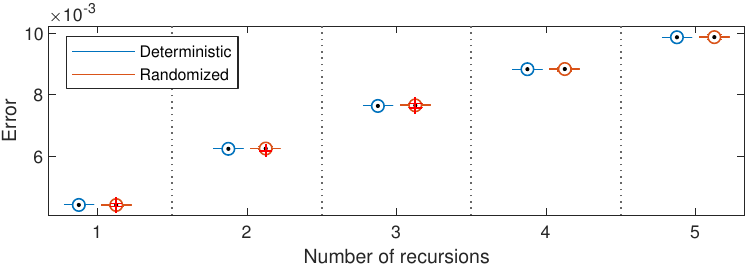}
	\includegraphics[width=1\textwidth]{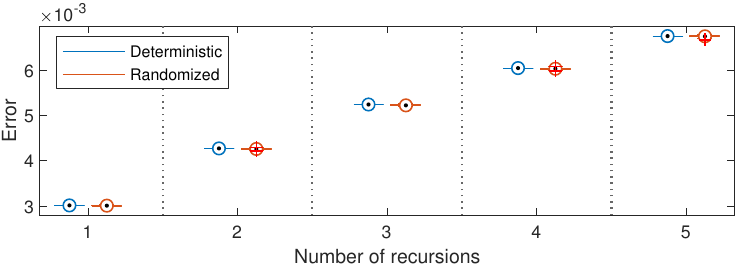}
	\captionof{figure}{(Five subplots above) Error for determinstic ABC compared to the error for the randomized counterpart, over 100 realizations of the randomized algorithm. $\Abf, \Bbf \in \Rb^{320 \times 320}$ are \textbf{Gaussian}, and each subplot corresponds to one realization of the pair $(\Abf,\Bbf)$.}
	\label{fig:S-experiment3-mat-type-normal}
\end{minipage}
\hfill
\begin{minipage}[ht!]{.47\textwidth}
	\includegraphics[width=1\textwidth]{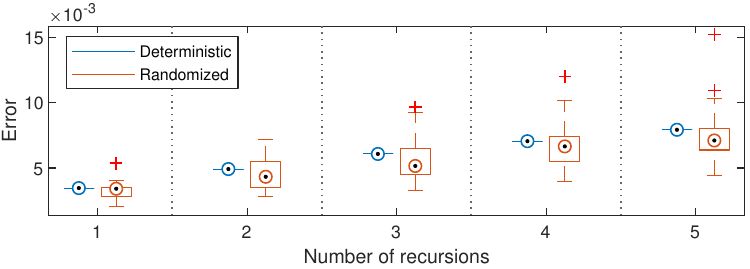}
	\includegraphics[width=1\textwidth]{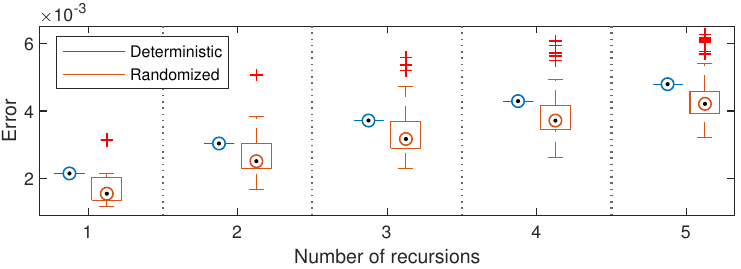}
	\includegraphics[width=1\textwidth]{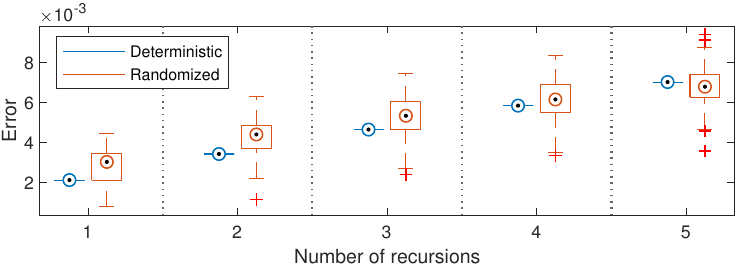}
	\includegraphics[width=1\textwidth]{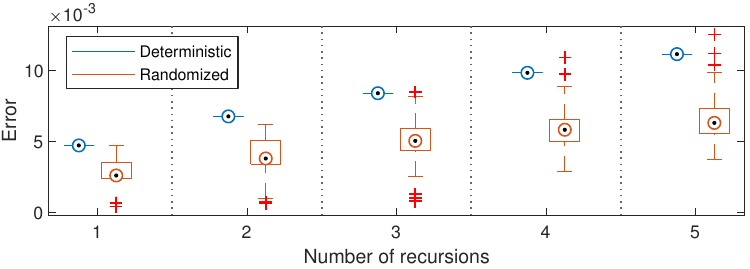}
	\includegraphics[width=1\textwidth]{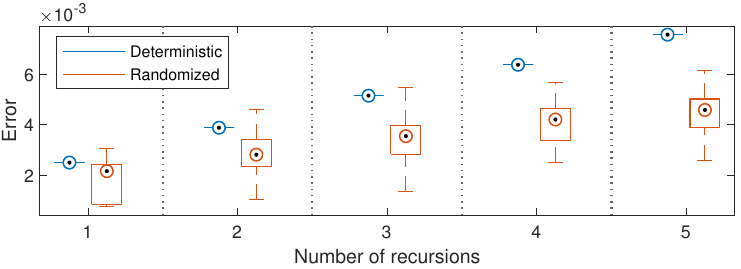}
	\captionof{figure}{(Five subplots above) Error for determinstic ABC compared to the error for the randomized counterpart, over 100 realizations of the randomized algorithm. $\Abf, \Bbf \in \Rb^{320 \times 320}$ are \textbf{uniform}, and each subplot corresponds to one realization of the pair $(\Abf,\Bbf)$.}
	\label{fig:S-experiment3-mat-type-uniform}
\end{minipage}

\begin{minipage}[ht!]{.47\textwidth}
	\includegraphics[width=1\textwidth]{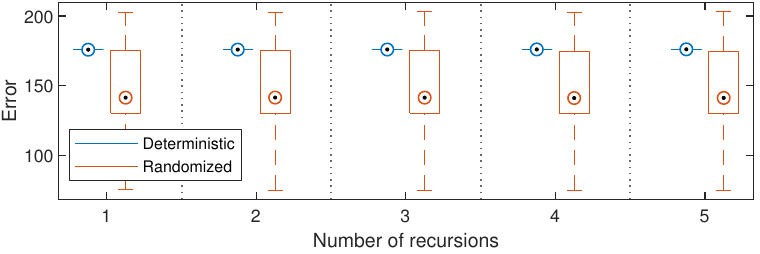}
	\includegraphics[width=1\textwidth]{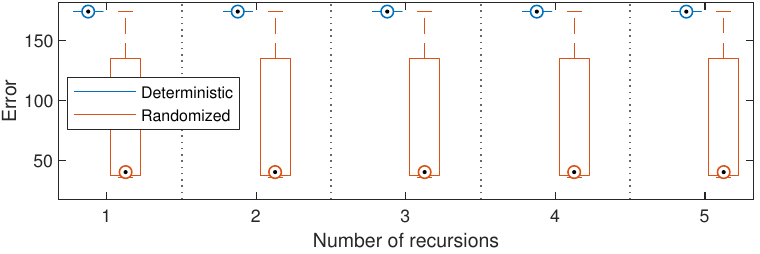}
	\includegraphics[width=1\textwidth]{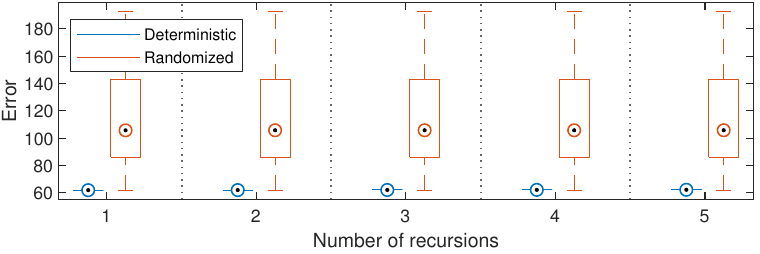}
	\includegraphics[width=1\textwidth]{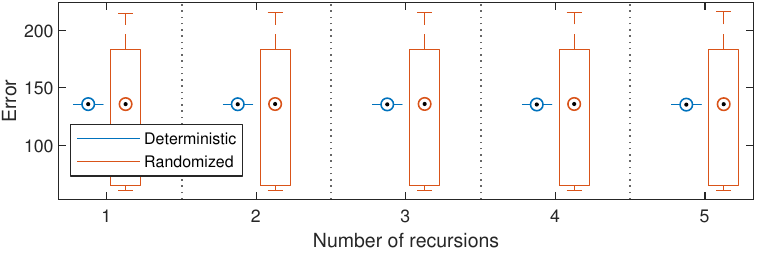}
	\includegraphics[width=1\textwidth]{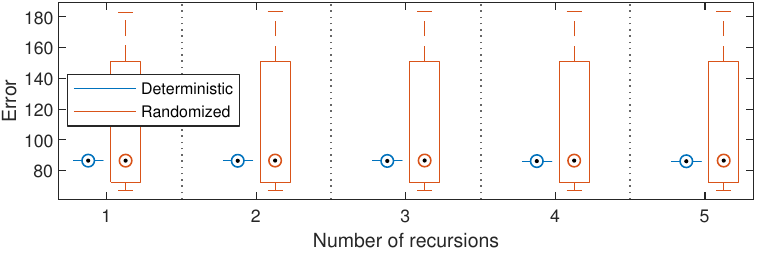}
	\captionof{figure}{(Five subplots above) Error for determinstic ABC compared to the error for the randomized counterpart, over 100 realizations of the randomized algorithm. $\Abf, \Bbf \in \Rb^{320 \times 320}$ are \textbf{type 1 adversarial}, and each subplot corresponds to one realization of the pair $(\Abf,\Bbf)$.}
	\label{fig:S-experiment3-mat-type-adversarial-1}
\end{minipage}
\hfill
\begin{minipage}[ht!]{.47\textwidth}
	\includegraphics[width=1\textwidth]{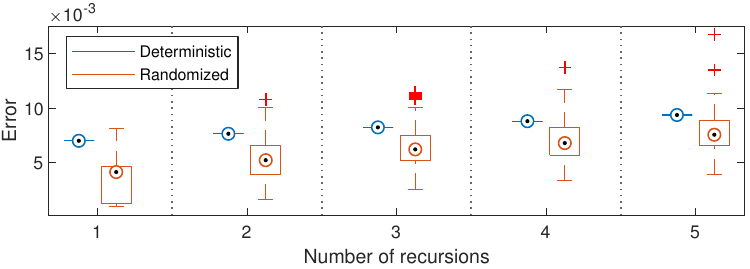}
	\includegraphics[width=1\textwidth]{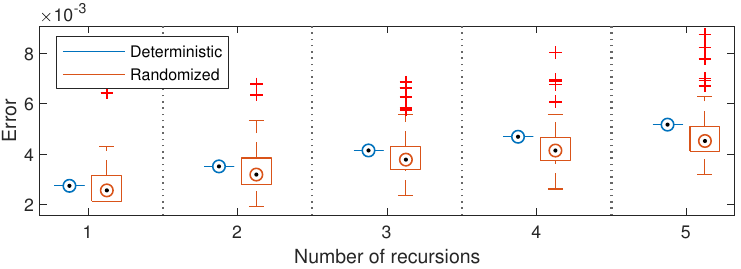}
	\includegraphics[width=1\textwidth]{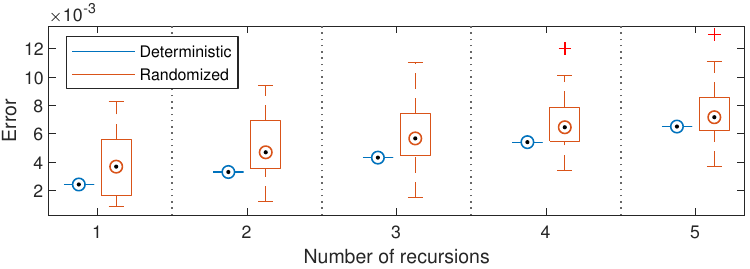}
	\includegraphics[width=1\textwidth]{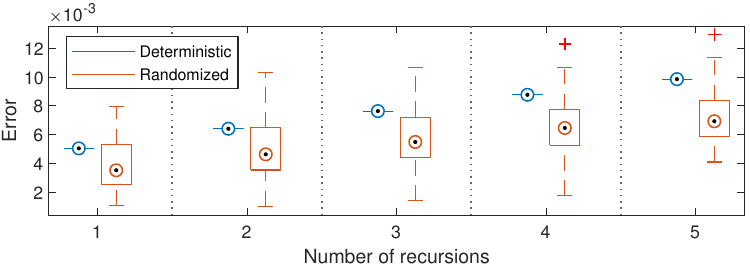}
	\includegraphics[width=1\textwidth]{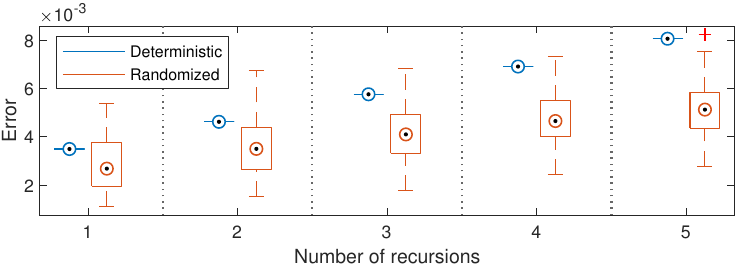}
	\captionof{figure}{(Five subplots above) Error for determinstic ABC compared to the error for the randomized counterpart, over 100 realizations of the randomized algorithm. $\Abf, \Bbf \in \Rb^{320 \times 320}$ are \textbf{type 2 adversarial}, and each subplot corresponds to one realization of the pair $(\Abf,\Bbf)$.}
	\label{fig:S-experiment3-mat-type-adversarial-2}
\end{minipage}

\begin{minipage}[ht!]{.47\textwidth}
	\includegraphics[width=1\textwidth]{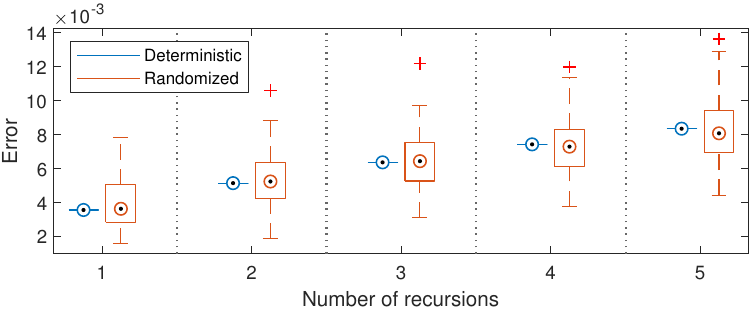}
	\includegraphics[width=1\textwidth]{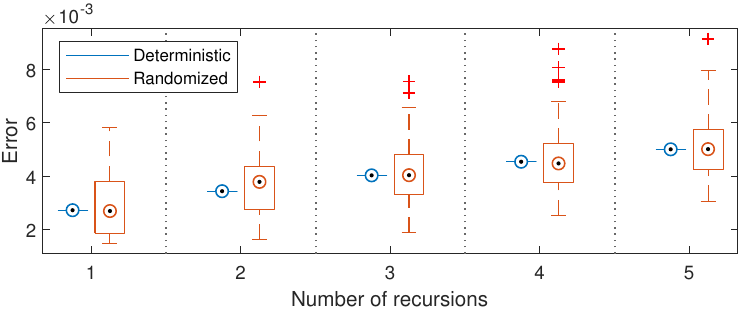}
	\includegraphics[width=1\textwidth]{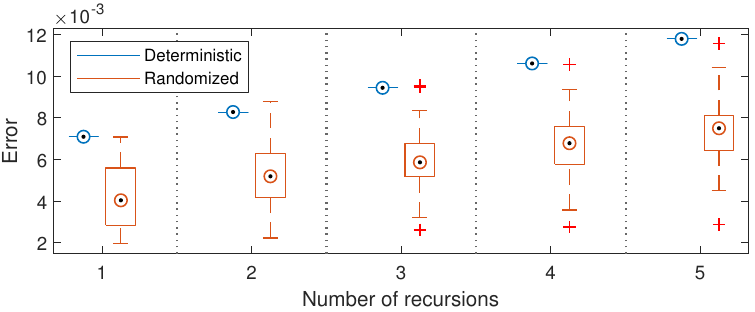}
	\includegraphics[width=1\textwidth]{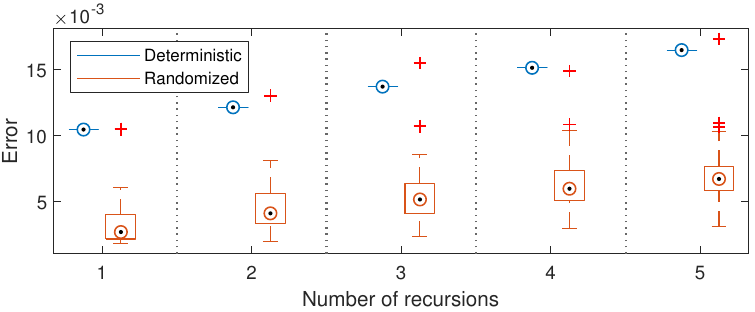}
	\includegraphics[width=1\textwidth]{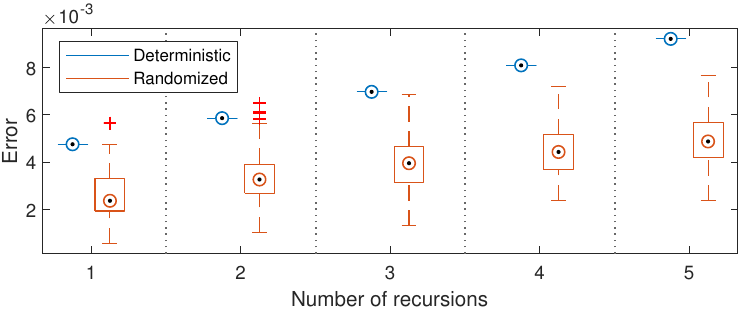}
	\captionof{figure}{(Five subplots above) Error for determinstic ABC compared to the error for the randomized counterpart, over 100 realizations of the randomized algorithm. $\Abf, \Bbf \in \Rb^{320 \times 320}$ are \textbf{type 3 adversarial}, and each subplot corresponds to one realization of the pair $(\Abf,\Bbf)$.}
	\label{fig:S-experiment3-mat-type-adversarial-3}
\end{minipage}
\hfill
\begin{minipage}[ht!]{.47\textwidth}
	\includegraphics[width=1\textwidth]{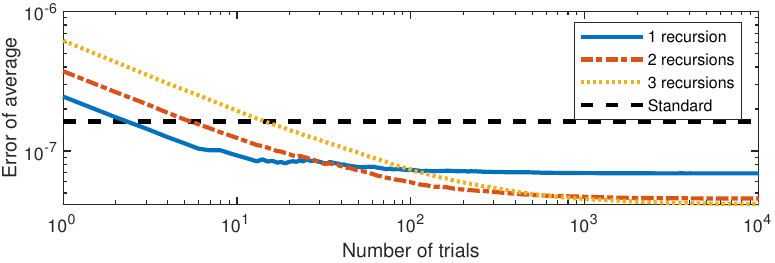}
	\includegraphics[width=1\textwidth]{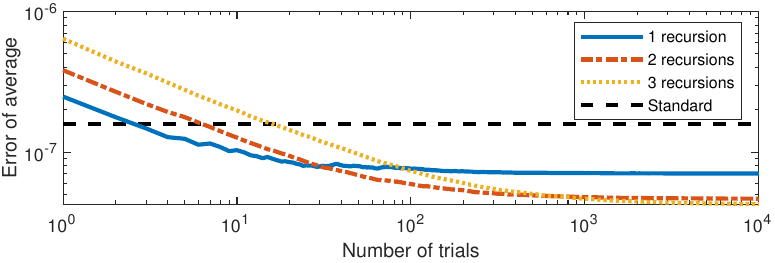}
	\includegraphics[width=1\textwidth]{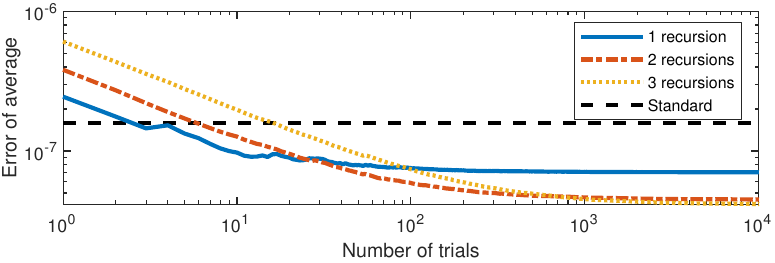}
	\includegraphics[width=1\textwidth]{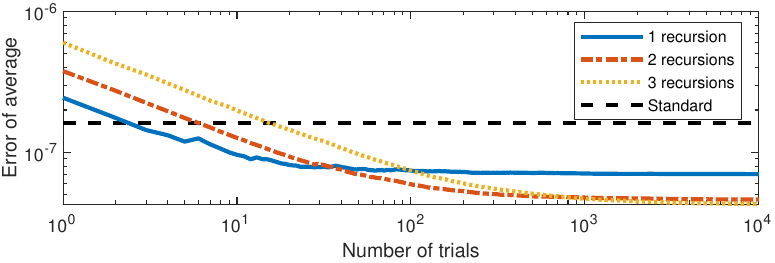}
	\includegraphics[width=1\textwidth]{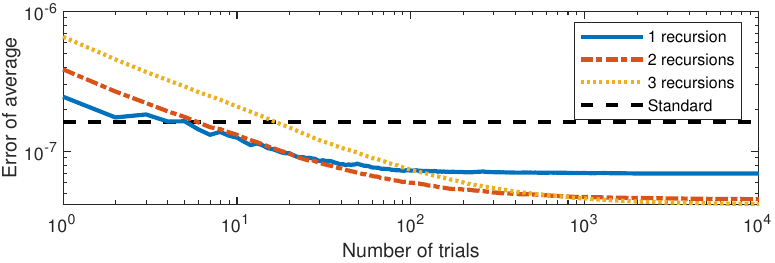}
	\captionof{figure}{(Five subplots above) Error for average of randomized EBC compared to the standard algorithm in single precision floating point arithmetic. $\Abf, \Bbf \in \Rb^{80 \times 80}$ are \textbf{Gaussian}, and each subplot corresponds to one realization of the pair $(\Abf,\Bbf)$.}
	\label{fig:S-experiment4-mat-type-normal}
\end{minipage}

\begin{minipage}[ht!]{.47\textwidth}
	\includegraphics[width=1\textwidth]{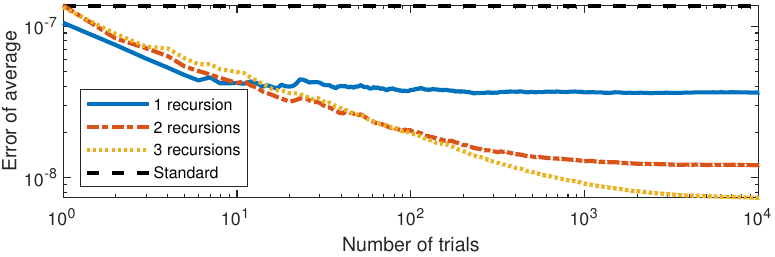}
	\includegraphics[width=1\textwidth]{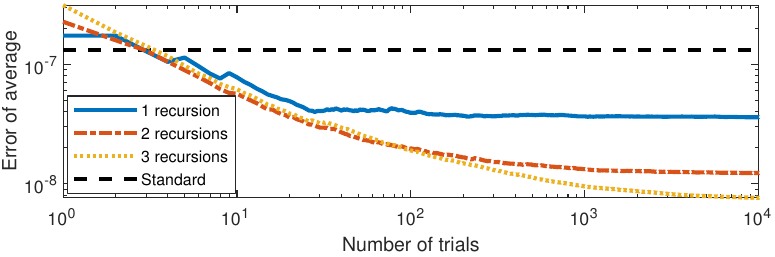}
	\includegraphics[width=1\textwidth]{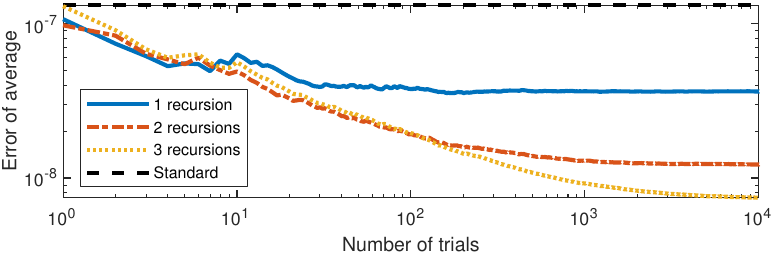}
	\includegraphics[width=1\textwidth]{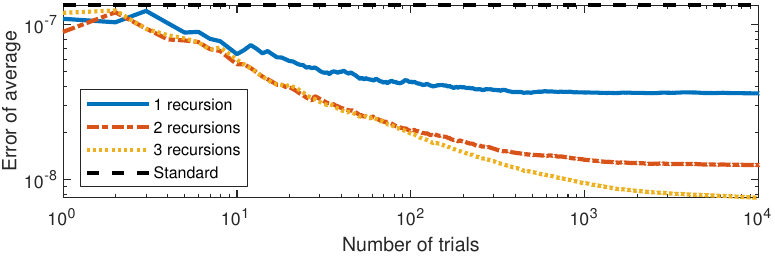}
	\includegraphics[width=1\textwidth]{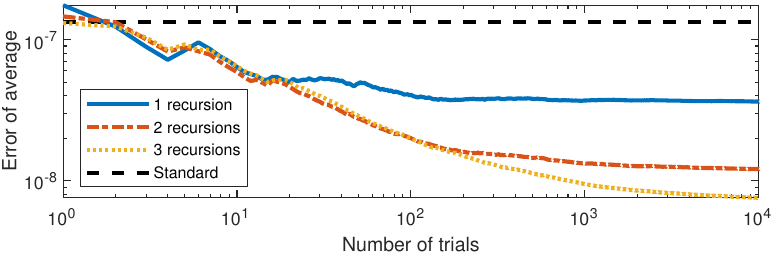}
	\captionof{figure}{(Five subplots above) Error for average of randomized EBC compared to the standard algorithm in single precision floating point arithmetic. $\Abf, \Bbf \in \Rb^{80 \times 80}$ are \textbf{uniform}, and each subplot corresponds to one realization of the pair $(\Abf,\Bbf)$.}
	\label{fig:S-experiment4-mat-type-uniform}
\end{minipage}
\hfill
\begin{minipage}[ht!]{.47\textwidth}
	\includegraphics[width=1\textwidth]{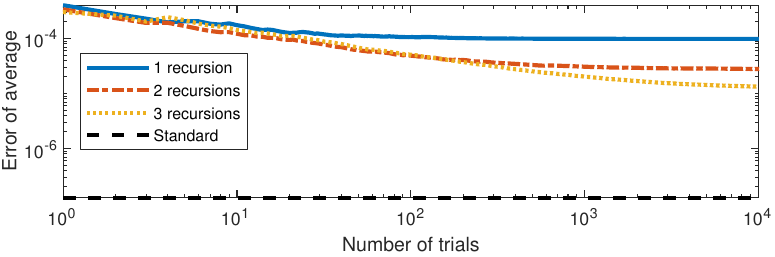}
	\includegraphics[width=1\textwidth]{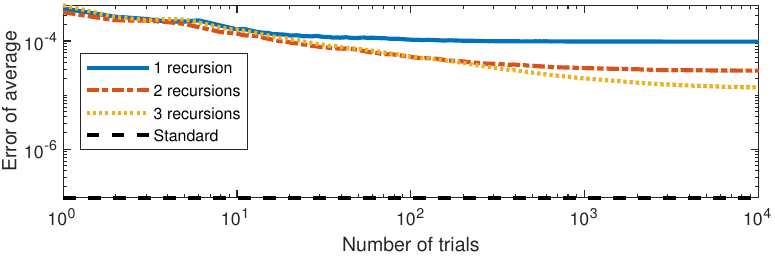}
	\includegraphics[width=1\textwidth]{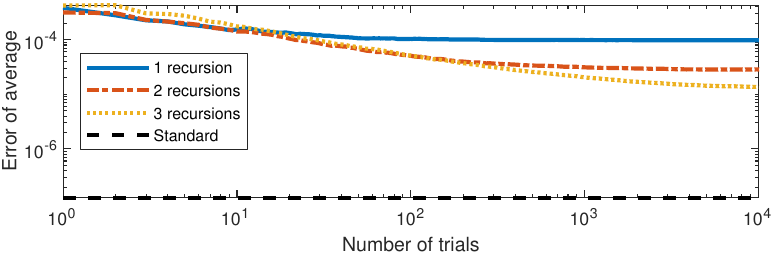}
	\includegraphics[width=1\textwidth]{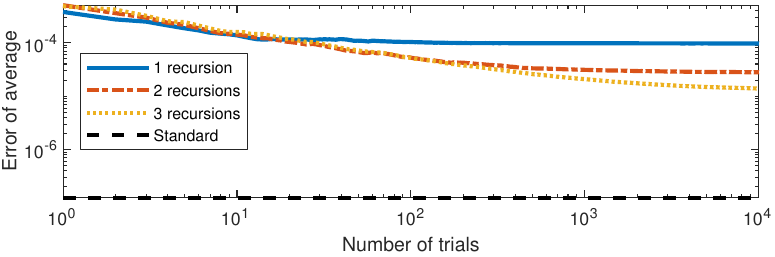}
	\includegraphics[width=1\textwidth]{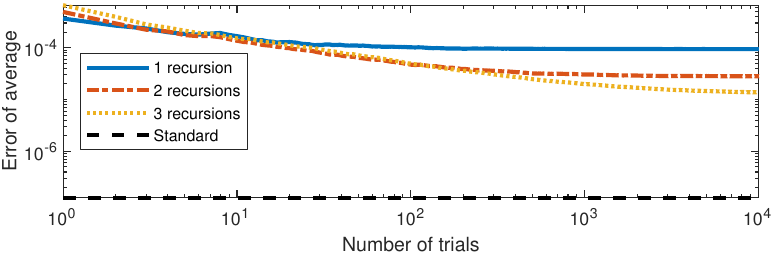}
	\captionof{figure}{(Five subplots above) Error for average of randomized EBC compared to the standard algorithm in single precision floating point arithmetic. $\Abf, \Bbf \in \Rb^{80 \times 80}$ are \textbf{type 1 adversarial}, and each subplot corresponds to one realization of the pair $(\Abf,\Bbf)$.}
	\label{fig:S-experiment4-mat-type-adversarial-1}
\end{minipage}

\begin{minipage}[ht!]{.47\textwidth}
	\includegraphics[width=1\textwidth]{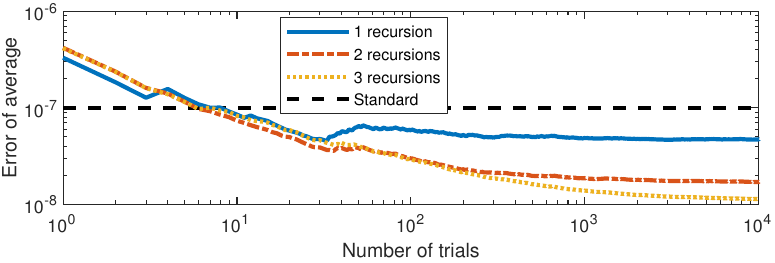}
	\includegraphics[width=1\textwidth]{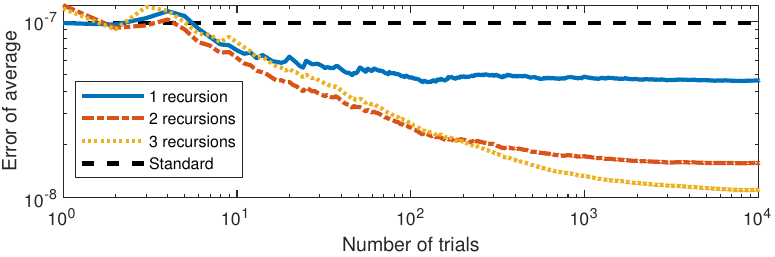}
	\includegraphics[width=1\textwidth]{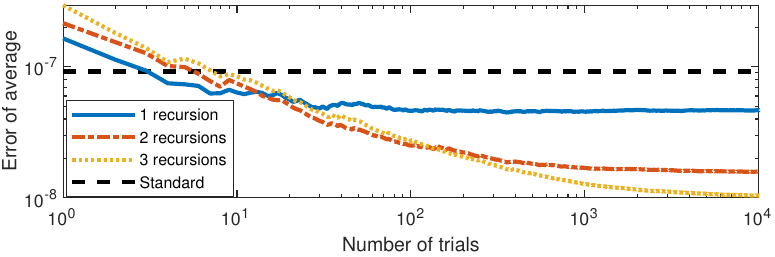}
	\includegraphics[width=1\textwidth]{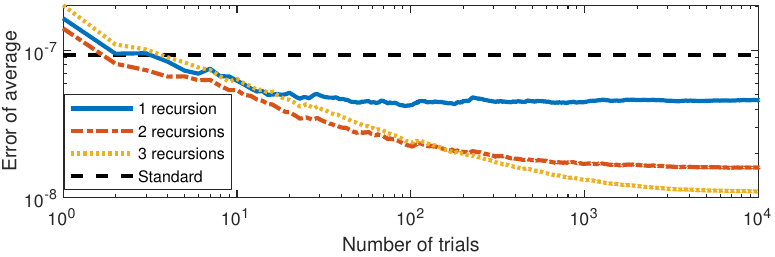}
	\includegraphics[width=1\textwidth]{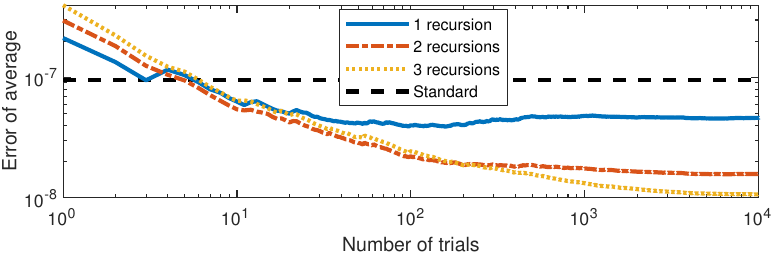}
	\captionof{figure}{(Five subplots above) Error for average of randomized EBC compared to the standard algorithm in single precision floating point arithmetic. $\Abf, \Bbf \in \Rb^{80 \times 80}$ are \textbf{type 2 adversarial}, and each subplot corresponds to one realization of the pair $(\Abf,\Bbf)$.}
	\label{fig:S-experiment4-mat-type-adversarial-2}
\end{minipage}
\hfill
\begin{minipage}[ht!]{.47\textwidth}
	\includegraphics[width=1\textwidth]{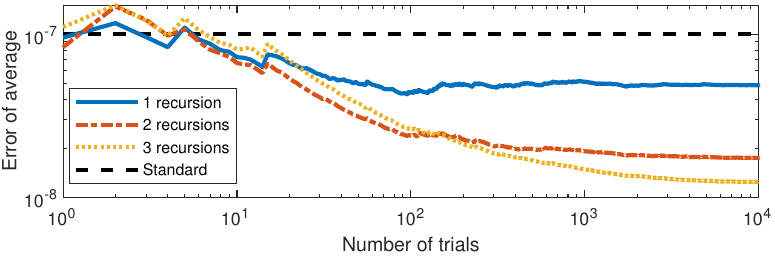}
	\includegraphics[width=1\textwidth]{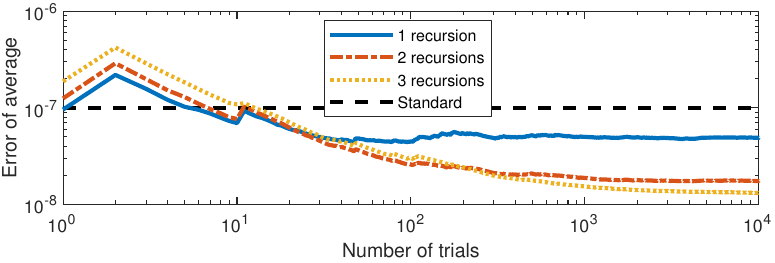}
	\includegraphics[width=1\textwidth]{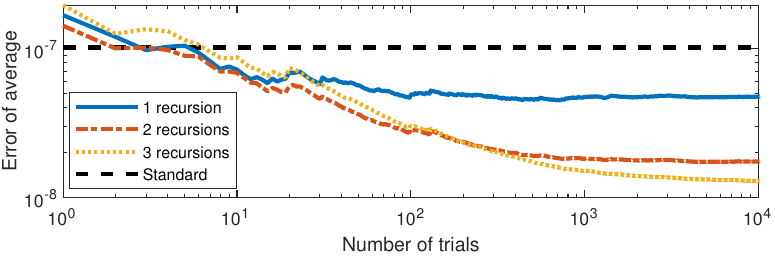}
	\includegraphics[width=1\textwidth]{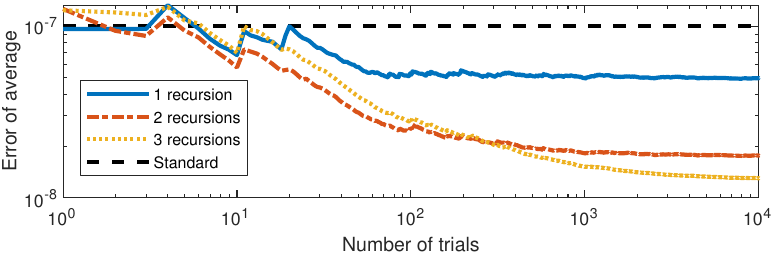}
	\includegraphics[width=1\textwidth]{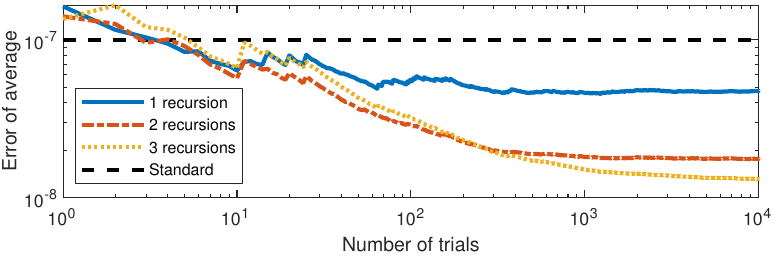}
	\captionof{figure}{(Five subplots above) Error for average of randomized EBC compared to the standard algorithm in single precision floating point arithmetic. $\Abf, \Bbf \in \Rb^{80 \times 80}$ are \textbf{type 3 adversarial}, and each subplot corresponds to one realization of the pair $(\Abf,\Bbf)$.}
	\label{fig:S-experiment4-mat-type-adversarial-3}
\end{minipage}

\begin{minipage}[ht!]{.47\textwidth}
	\includegraphics[width=1\textwidth]{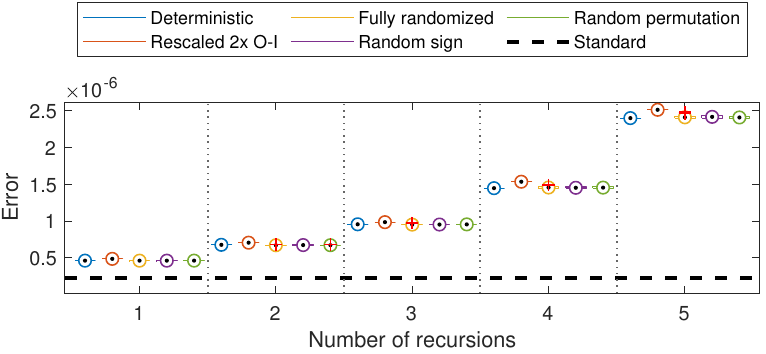}
	\includegraphics[width=1\textwidth]{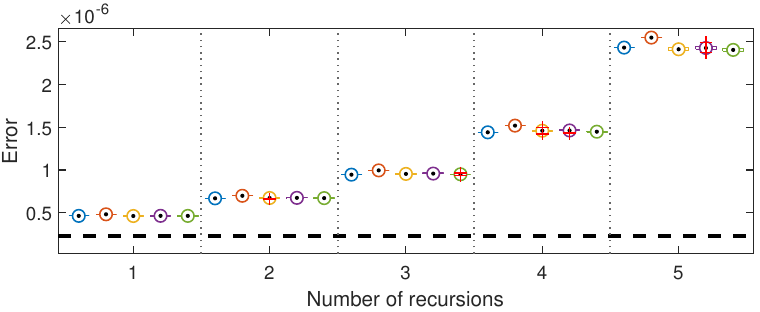}
	\includegraphics[width=1\textwidth]{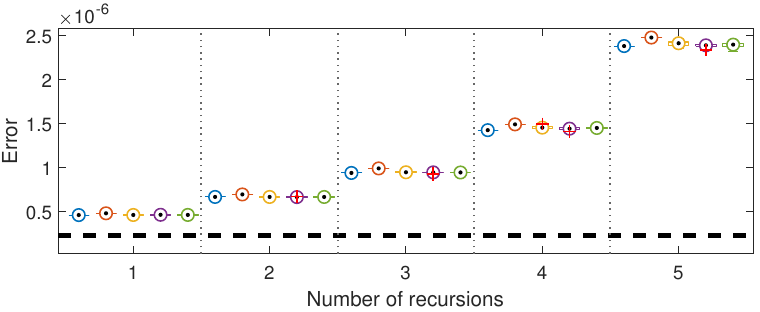}
	\includegraphics[width=1\textwidth]{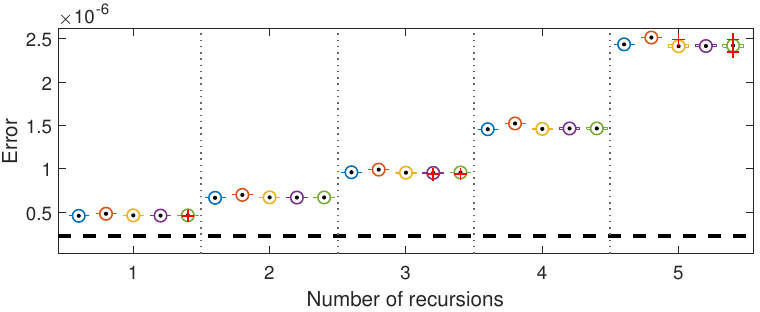}
	\includegraphics[width=1\textwidth]{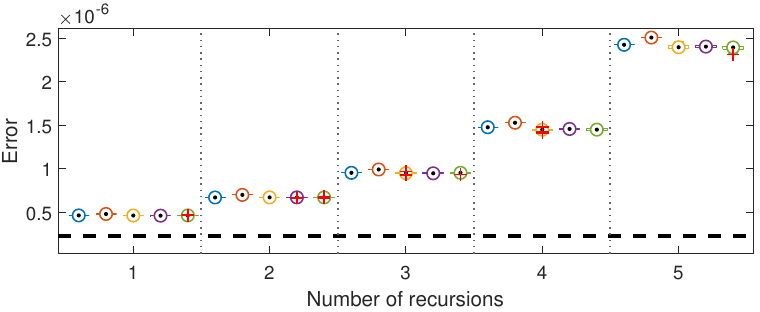}
	\captionof{figure}{(Five subplots above) Error for different variants of the EBC in single precision floating point arithmetic, over 100 realizations of the randomized algorithms. $\Abf, \Bbf \in \Rb^{320 \times 320}$ are \textbf{Gaussian}, and each subplot corresponds to one realization of the pair $(\Abf, \Bbf)$.}
	\label{fig:S-experiment5-mat-type-normal}
\end{minipage}
\hfill
\begin{minipage}[ht!]{.47\textwidth}
	\includegraphics[width=1\textwidth]{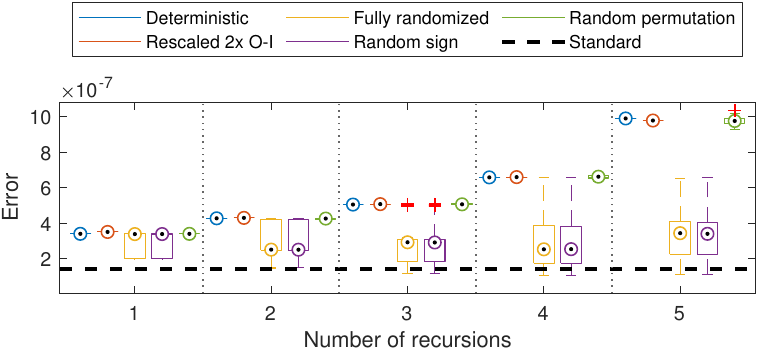}
	\includegraphics[width=1\textwidth]{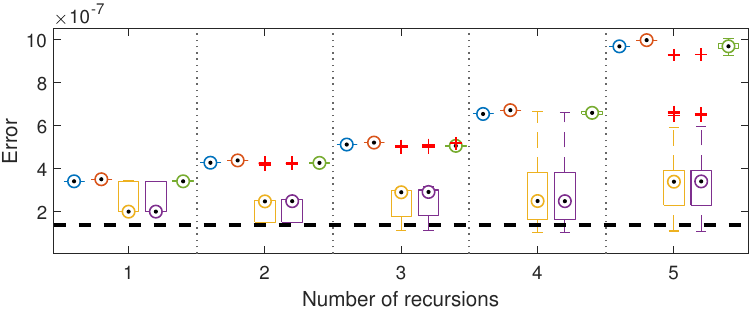}
	\includegraphics[width=1\textwidth]{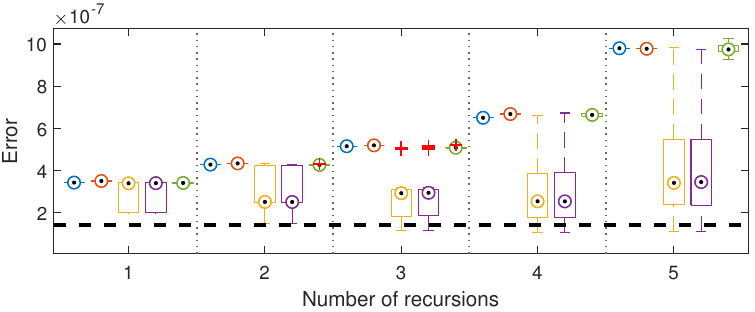}
	\includegraphics[width=1\textwidth]{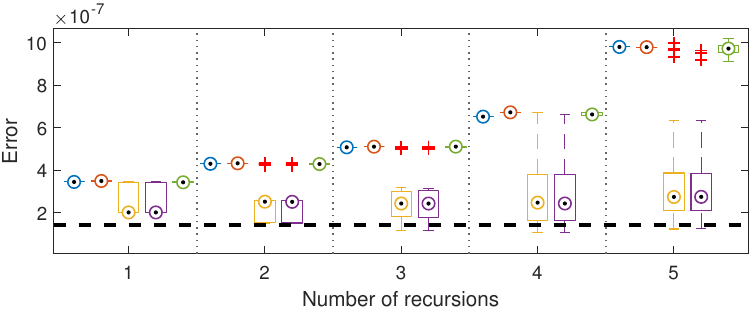}
	\includegraphics[width=1\textwidth]{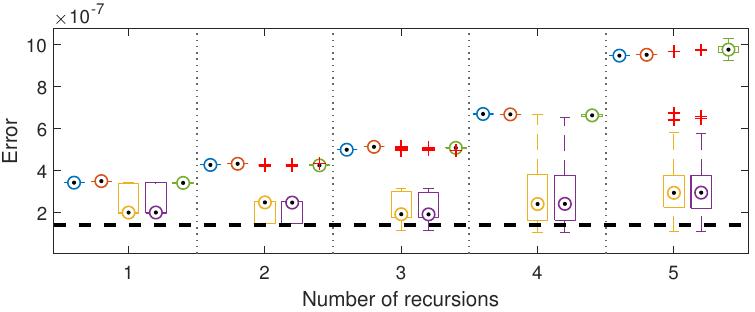}
	\captionof{figure}{(Five subplots above) Error for different variants of the EBC in single precision floating point arithmetic, over 100 realizations of the randomized algorithms. $\Abf, \Bbf \in \Rb^{320 \times 320}$ are \textbf{uniform}, and each subplot corresponds to one realization of the pair $(\Abf, \Bbf)$.}
	\label{fig:S-experiment5-mat-type-uniform}
\end{minipage}

\begin{minipage}[ht!]{.47\textwidth}
	\includegraphics[width=1\textwidth]{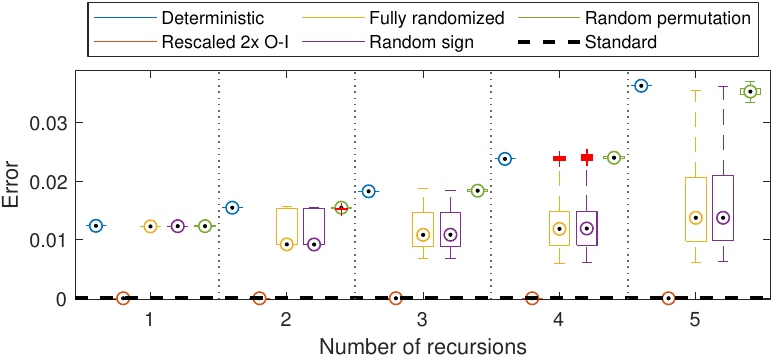}
	\includegraphics[width=1\textwidth]{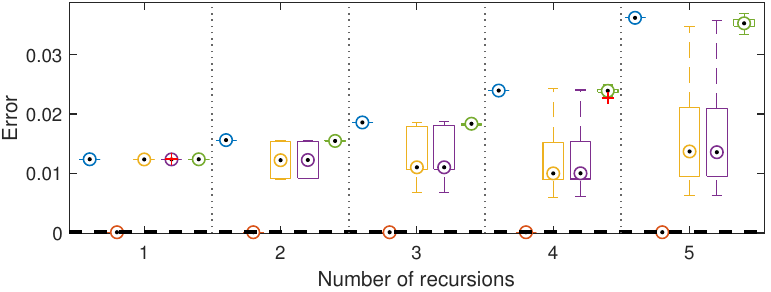}
	\includegraphics[width=1\textwidth]{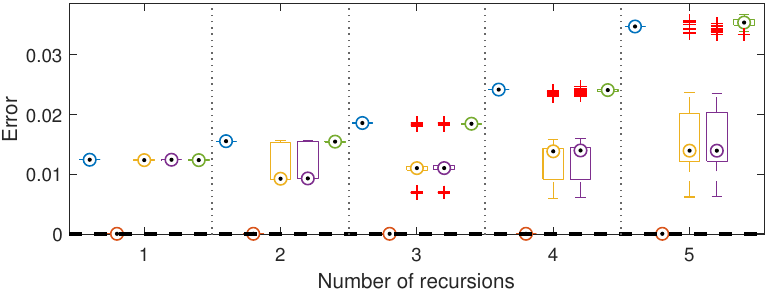}
	\includegraphics[width=1\textwidth]{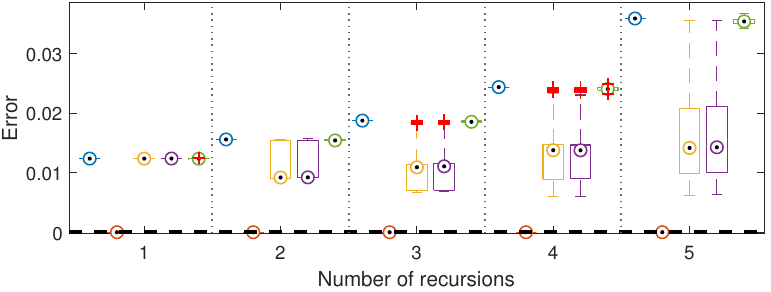}
	\includegraphics[width=1\textwidth]{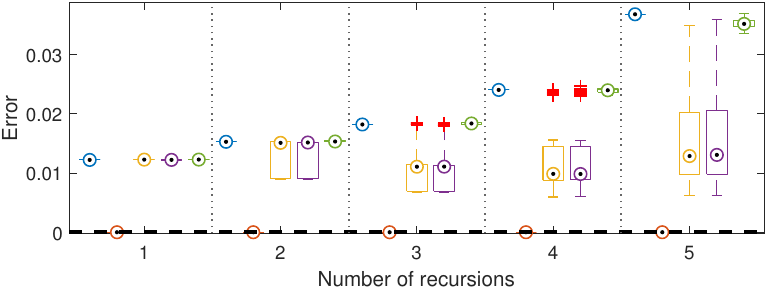}
	\captionof{figure}{(Five subplots above) Error for different variants of the EBC in single precision floating point arithmetic, over 100 realizations of the randomized algorithms. $\Abf, \Bbf \in \Rb^{320 \times 320}$ are \textbf{type 1 adversarial}, and each subplot corresponds to one realization of the pair $(\Abf, \Bbf)$.}
	\label{fig:S-experiment5-mat-type-adversarial-1}
\end{minipage}
\hfill
\begin{minipage}[ht!]{.47\textwidth}
	\includegraphics[width=1\textwidth]{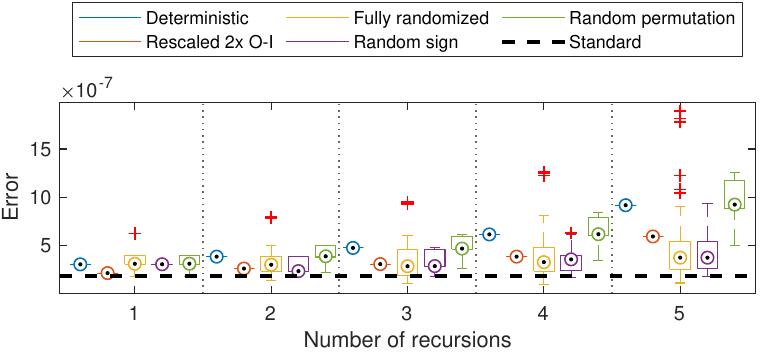}
	\includegraphics[width=1\textwidth]{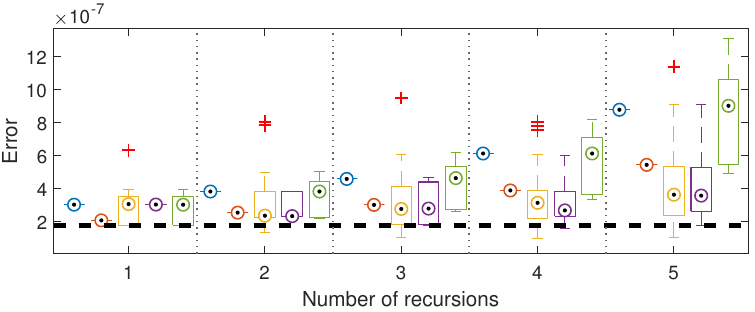}
	\includegraphics[width=1\textwidth]{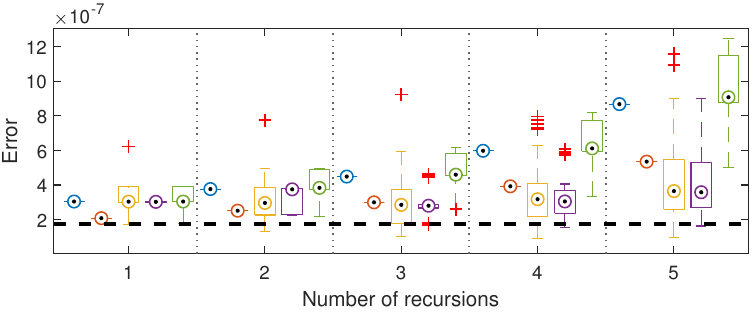}
	\includegraphics[width=1\textwidth]{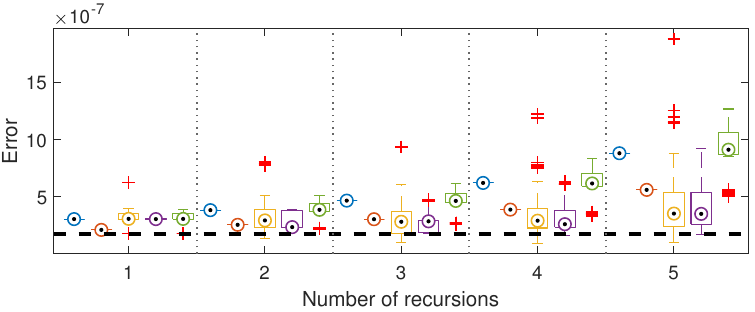}
	\includegraphics[width=1\textwidth]{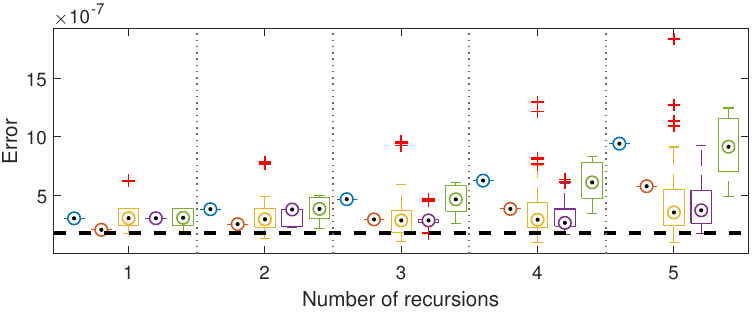}
	\captionof{figure}{(Five subplots above) Error for different variants of the EBC in single precision floating point arithmetic, over 100 realizations of the randomized algorithms. $\Abf, \Bbf \in \Rb^{320 \times 320}$ are \textbf{type 2 adversarial}, and each subplot corresponds to one realization of the pair $(\Abf, \Bbf)$.}
	\label{fig:S-experiment5-mat-type-adversarial-2}
\end{minipage}

\begin{minipage}[ht!]{.47\textwidth}
	\includegraphics[width=1\textwidth]{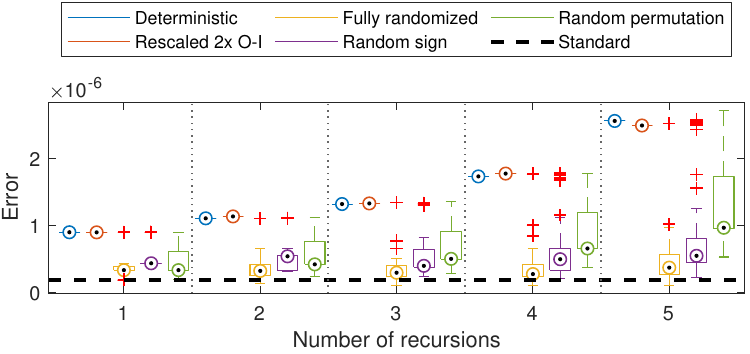}
	\includegraphics[width=1\textwidth]{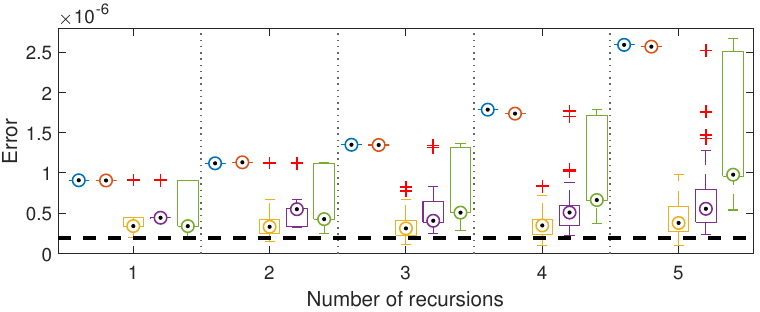}
	\includegraphics[width=1\textwidth]{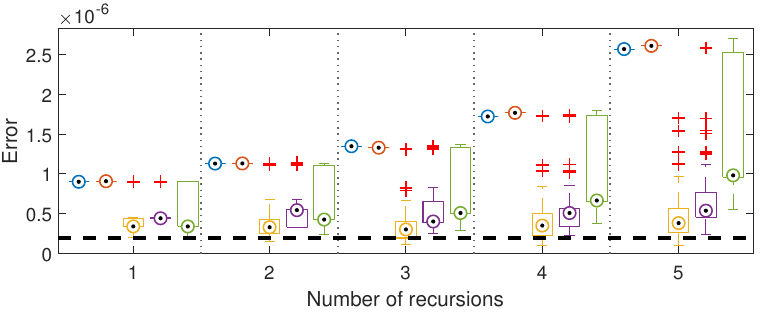}
	\includegraphics[width=1\textwidth]{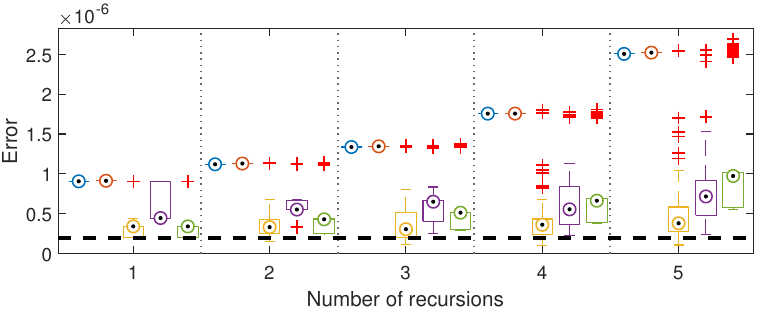}
	\includegraphics[width=1\textwidth]{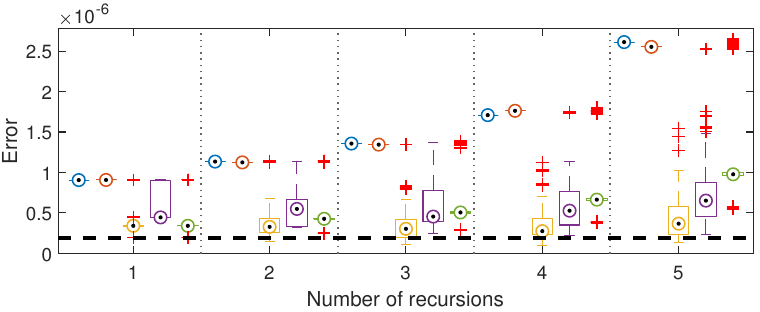}
	\captionof{figure}{(Five subplots above) Error for different variants of the EBC in single precision floating point arithmetic, over 100 realizations of the randomized algorithms. $\Abf, \Bbf \in \Rb^{320 \times 320}$ are \textbf{type 3 adversarial}, and each subplot corresponds to one realization of the pair $(\Abf, \Bbf)$.}
	\label{fig:S-experiment5-mat-type-adversarial-3}
\end{minipage}

\end{document}